\documentclass[11pt]{article}
\usepackage[utf8]{inputenc}
\usepackage{color}
\usepackage{graphicx}

\usepackage{amsthm}
\usepackage{amssymb}
\usepackage{amsmath}
\usepackage{float}
\usepackage{bm}
\usepackage{subcaption}
\newtheorem{theorem}{Theorem}

\newtheorem{lemma}{Lemma}

\usepackage{graphicx}

\graphicspath{{figs/}{DNA_simulation/3layers/}}

\begin{document}
\title{Knotted DNA Configurations in Bacteriophage Capsids: A Liquid Crystal Theory Approach\thanks{P.L. was partially supported by NSF grant DMS-2318053. Z.W., T.C. and J.A. were partially supported by NSF grant DMS-2318052 and DMS-1817156. MC.C. was partially supported by NSF grant DMS-2318051. M.V. was partially supported by NSF grant DMS-2054347 and DMS-1817156.}}

\author{Pei Liu\thanks{Department of Mathematics and Systems Engineering, Florida Institute of Technology, Melbourne, FL, 32901.}
\and Zhijie Wang\thanks{Graduate Program in Applied Mathematics, University of California, Davis, CA 95616.}
\and Tamara Christiani\thanks{Integrative Genetics and Genomics Graduate Group, University of California, Davis, CA 95616.}
\and Mariel Vazquez
  \thanks{
Department of Mathematics and Department of Microbiology and Molecular Genetics, University of California, Davis, CA 95616.
}
\and M. Carme Calderer
 \thanks{
 School of Mathematics, University of Minnesota, Minneapolis, MN 55442 (calde014@umn.edu)
}
\and Javier Arsuaga
\thanks{Department of Molecular and Cellular Biology and Department of Mathematics, University of California, Davis, CA 95616 (jarsuaga@ucdavis.edu)
}}

\date{}
\maketitle

\begin{abstract}
Bacteriophages, viruses that infect bacteria, store their micron long DNA inside an icosahedral capsid with a typical diameter of  40 nm to 100 nm. Consistent with experimental observations, such confinement conditions induce an arrangement of DNA that corresponds to a hexagonal chromonic liquid-crystalline phase, and increase the topological complexity of the genome in the form of knots. A mathematical model that implements a chromonic liquid-crystalline phase and that captures the changes in topology has been lacking.  We adopt a mathematical model that represents  the viral DNA as a pair  of a vector field and a line. The vector field is a minimizer of the total Oseen–Frank energy for nematic liquid crystals under chromonic constraints, while the line is identified with the tangent to the field at selected locations, representing the central axis of the DNA molecule. The fact that the Oseen–Frank functional assigns infinite energy to topological defects (point defects in two dimensions and line defects in three dimensions) precludes the presence of singularities and, in particular, of knot structures.  To address this issue, we begin with the optimal vector field and helical line, and propose a new algorithm to introduce knots through stochastic perturbations associated with splay and twist deformations, modeled by means of a Langevin system. We conclude by comparing knot distributions generated by the model and by interpreting them in the context of previously published experimental results. Altogether, this work relies on the synergy of modeling, analysis and computation in the study of viral DNA organization in capsids. 
\end{abstract}

\section{Introduction}

Knots are ubiquitous in macromolecules and have been widely observed in both DNA and proteins (see \cite{tubiana2024topology} for a recent review). DNA knots may arise through several mechanisms, including site-specific recombination reactions \cite{colloms1997topological,crisona1999topological}, the action of topoisomerases \cite{lopez2012topo,valdes2019transcriptional,Wasserman1991}, and random cyclization experiments in either free solution \cite{rybenkov1993probability,shaw1993knotting} or confined environments \cite{Arsuaga2002b,liu1981knotted,Liu1981}. Mathematical analyses of these data have helped elucidate the role of enzymes \cite{moore2020recent,shimokawa2013ftsk,valencia2011predicting} and deepened our understanding of the DNA duplex and of the statistical behavior of DNA knots \cite{rybenkov1993probability,tesi1994knotting,beaton2018characterising,beaton2024first,micheletti2011polymers,Micheletti2008}.

In the bacterial virus, bacteriophage P4, one single DNA molecule is packed in an icosahedral protein structure called capsid. It has been experimentally observed that most DNA molecules extracted from P4 capsids are knotted \cite{Arsuaga2002b,liu1981knotted,Wolfson1985}. Further analyses of the P4 knots have revealed that knotting is mostly driven by confinement, and modulated by the liquid crystal structure of DNA and the packing reaction \cite{arsuaga2002knotting, Arsuaga2005,arsuaga2008,marenduzzo2009dna}. The study of P4 knots has often relied on mathematical models of random knotting \cite{Arsuaga2005,arsuaga2007sampling,diao1994random,even2017models,ishihara2017bounds,liu2020characterizing,ziegler2012computational}.  These models however fail to implement a reliable representation of  the DNA molecule. 

In P4 and other bacteriophages, the DNA molecule is densely packed within the capsid whose radius is at least one order of magnitude smaller than the length of packed genome. The resulting DNA concentration can reach up to 800 mg/ml \cite{Kellenberger1986}, a regime in which in vitro DNA has been shown to form a hexagonal chromonic liquid crystalline phase \cite{livolant1989highly,Rill1986,Strzelecka1988,leforestier2009structure,Lepault1987,livolant1984cholesteric,pelta1996dna}. Although DNA organization varies depending on concentration, capsid shape \cite{petrov2007conformation}, ionic conditions \cite{leforestier2009structure}, and the packaging mechanism \cite{cruz2020quantitative}, cryo-EM imaging reveals well-ordered coaxial layers of DNA near the capsid wall, consistent with the liquid crystal phase of DNA which transition into a disordered configuration near the center of the capsid \cite{Comolli2008,Lander2006}. 

Let us now introduce some basic concepts and liquid-crystal terminology. Small-molecule liquid crystals are  divided into two classes: calamitic and discotic. Calamitic mesogens are rigid, rod-like molecules whose long axes tend to align along a preferred direction, or optic axis.  In discotic liquid crystals, disk-like molecules stack face-to-face into columns, which also tend to follow a preferential axial alignment. The nematic phase is characterized by orientational order—either of the rod-like molecules in calamitics or of the columnar stacks in discotics. In the discotic case, under lower temperatures (thermotropic liquid crystals) or at higher concentrations (lyotropic systems), these columns may further organize into hexagonal arrays, forming the columnar hexagonal phase. Chromonic liquid crystals, also referred to as the liquid crystals of life, represent a special subclass of lyotropic columnar systems in water, with the hexagonal phase corresponding to their highest degree of order. 
These concepts also extend to semiflexible polymers, and in particular to DNA, whose intrinsic stiffness and anisotropy enable it to exhibit analogous liquid-crystalline behavior under appropriate conditions.

To the best of our knowledge, the first observation of condensed DNA forming liquid crystal phases can be traced to the (1978) article by Livolant and Bouligand \cite{LivolantBouligand1978}. In such an article, the authors used optical microscopy to study the twisted, cholesteric organization of in-vivo dinoflagellate chromosomes in connection with DNA self-assembly. In 1984, Livolant explicitly demonstrated that in-vitro solutions of 
concentrated DNA also form a cholesteric liquid-crystal phase \cite{Livolant1984a}. The landmark observation,  based on X-ray diffraction studies, of   highly condensed DNA  forming a hexagonal chromonic phase   was later  reported by Livolant et al. in 1989 \cite{livolant1989highly}.  Later studies by Leforestier and Livolant further confirmed the earlier X-ray finding by cryogenic electron microscopy \cite{LeforestierLivolant1991, Leforestier1993}. Important advances in liquid crystal science also occurred in the late part of the 
 20th century, including the discovery of pharmaceutic drug compounds,  such as the antiashtmatic disodium cromoglycate (DSCG), which together with water-soluble  food dyes such as Sunset Yellow, were  identified as hexagonal chromonic liquid crystals. 
   However, the chromonic denomination of  such liquid crystals, attributed to John Lyndon in 2004, was based on their analogy with configurations of DNA condensates, but  it was not until  later that found its application in designating the DNA phases \cite{Lydon2004}.

The Oseen-Frank energy of a liquid crystal provides the classical framework for describing orientational order in the nematic phase. The theory models the distortion energy of a unit vector field $\vec{n}$, known as the director field, which represents the local tangent direction of DNA filaments. The energy penalizes deviations from a uniform alignment, with separate contributions from splay, twist, and bend distortions. The equilibrium configuration corresponds to a director field that minimizes the total distortion energy, subject to boundary and the unit length constraint $|\vec{n}|=1:$
\begin{eqnarray}
    2E_{OF}&=& \int_{\Omega} \left[ k_1 (\nabla \cdot \vec{n})^2 + k_2 (\vec{n} \cdot \nabla \times \vec{n})^2 + k_3 |\vec{n} \times \nabla \times \vec{n}|^2 \right] d\vec{x}\nonumber \\
    &+& \int_{\partial \Omega} (k_2+k_4)\big((\nabla\vec{n})\vec{n}-(\nabla\cdot\vec{n})\vec{n}\big) \cdot d\vec{S}.\label{S0}
\end{eqnarray}
In this expression,  $\Omega$ represents the domain occupied by the liquid crystal and $k_1, k_2, k_3$ are the splay, twist, bending Frank constants. The term with the coefficient     `$k_2+k_4$' is the saddle-splay term, a null-Lagrangian that accounts for  surface energy. Eq. \eqref{S0} has a unique minimizer provided that $k_1, k_2, k_3>0$, $k_2 \ge |k_4|$ and $2k_1 \ge k_2+k_4$ \cite{hardt1986existence}. 
 The Oseen-Frank energy has also been adopted to modeling the hexagonal chromonic phase by imposing a zero-divergence constraint, which favors energy minimizes with layered structures, while preventing strand crossings. In our setting, the capsid plays the confining role otherwise assigned to an additional surface energy,    in the case of free shapes \cite{koizumi2022toroidal, hiltner2021chromonic}.    
Motivated by these observations, we have developed an application of the theory to study the spooling of DNA in confined geometries \cite{walker2020liquid,walker2020fine,hiltner2021chromonic,liu2021ion,liu2022helical,ortiz2003}. 
 However, its energy functional penalizes defects with infinite cost, and therefore cannot capture knotted DNA configurations through deterministic minimization alone.

In this work, we build upon our previous studies \cite{liu2022helical,liu2025submitted} to find the helical energy minimizing vector field and the corresponding  unknotted tangent line. 
To further incorporate knotted structures, we introduce the effect of thermal fluctuations through  stochastic perturbations, via Langevin dynamics. This allows us to simulate ensembles of DNA configurations consistent with the underlying liquid crystal ordering while capturing topological complexity.  Another main departure from the previous work is the incorporation of a variable mass density of the DNA. This accounts for the fact that  DNA inside the capsid is neither uniformly distributed nor incompressible. Indeed, confinement leads to local variations in packing density. Consistent with  the incorporation of the  density is the addition of an entropic energy term and a chemical potential. At the molecular scale, DNA segments are semiflexible polymers with thermal fluctuations, with a residual configurational entropy associated with local alignment or disorder and density fluctuations; it accounts for the entropic cost of confinement.  The chemical potential controls the energetic cost of adding or removing DNA mass into the capsid.  

This article is structured as follows. In section 2, we  present the liquid crystal model of DNA organization.  This model extends our previous studies  by accounting for the variable DNA mass density, the associated entropic energy  and a DNA chemical potential. We also include the  saddle-splay term in the Oseen-Frank energy. 
We show that the resulting nonlinear, second order ordinary differential Euler--Lagrange equation for the helical angle $\psi$ has a unique minimizer. 
In section 3, we introduce the algorithms that generate random walks governed by Langevin equations and that produce knotting distributions on the minimizer structure.  In section 4, we numerically solve the Euler-Lagrange equations and analyze the structure of the minimizer according to available biological data and parameter selection. We generate random distributions of knotted configurations and compare them with published experimental knot distributions \cite{Arsuaga2005}. We conclude by discussing the implications of our study and possible extensions of this work. 

\section{A Continuum Model}\label{continuum_model} In this section, we consider the average configuration of the ordered DNA according to the liquid crystal theory, ignoring the disordered region at the center as well as the north and south caps. The DNA is modeled as following helical  trajectories on cylindrical surfaces bounded by a fixed inner $R_1$ and an outer radius $R_2$,  with $R_2>R_1>0$.  The local configuration  is described by a unit vector field $\vec{n}$, representing the tangent direction of the DNA filament, together with a scalar field $\rho$, which specifies  the local packing density. The ordered domain, denoted by $\Omega$, is taken to be a cylindrical shell of height $H$, with inner radius $R_1$ and outer radius $R_2$. In cylindrical coordinates,
 $$\Omega=\{(r, \theta, z): R_1\leq r\leq R_2,  0\leq \theta\leq 2\pi,  0\leq z\leq H\}.$$  

The Euler--Lagrange equations that describe the equilibrium configuration of the DNA within the ordered phase, are derived by minimizing the following energy,
\begin{eqnarray}
    E_{total} &=& \frac{1}{2}\int_\Omega \rho(\vec{x})\left[ K_1 (\nabla \cdot \vec{n})^2 + K_2 (\vec{n} \cdot \nabla \times \vec{n})^2 + K_3 |\vec{n} \times \nabla \times \vec{n}|^2 \right] d\vec{x}\nonumber \\
    &+& \frac{1}{2}\int_{\partial \Omega} \rho(\vec{x})(K_2+K_4)\big((\nabla\vec{n})\vec{n}-(\nabla\cdot\vec{n})\vec{n}\big) \cdot d\vec{S} + k_B T\int_\Omega \rho (\log \rho + \mu) d\vec{x}\label{S1}
\end{eqnarray}
Here $\rho K_1, \rho K_2, \rho K_3$ are the splay, twist, bending Frank coefficients, respectively. We assume that  the Frank coefficients scale  proportionally with the local density $\rho$, where $K_1$, $K_2$ and $K_3$ are material constants. The surface integral corresponds to the saddle-splay term in the Oseen--Frank theory, accounting for boundary contributions to the energy. The last term represents  the entropic contribution, with $k_B$ denoting the Boltzmann constant, $T$ being the temperature and $\mu$ the chemical potential of the DNA molecule. In the case that the total mass of ordered DNA is fixed, the constraint
\begin{equation}
\int_\Omega \rho(\vec{x}) d\vec{x} = N \label {mass-total}\end{equation}  holds and  the chemical potential $\mu$ corresponds to the Lagrange multiplier.

In order to represent the helical structure of the DNA, we introduce the unit vector $\vec{n}$ at the point $(r, \theta, z)$ to the helical curve,
\begin{equation}
    \vec{n}(r,\theta,z) =  \cos\psi \vec{e}_\theta+ \sin\psi\vec{e}_z, \quad \psi=\psi(r, \theta, z),\label{helical}
\end{equation}
with $\psi \in (0,\pi/2)$  representing the local helical angle. Furthermore, we look for minimizers that satisfy the constraint,
\begin{equation}
    \nabla \cdot \vec{n} = -\frac{\sin\psi}{r} \psi_\theta + \cos \psi \psi_z=0 \Longrightarrow \psi_z = \frac{\tan \psi}{r} \psi_\theta. \label{constraint}
\end{equation}
The zero-splay constraint in a liquid crystal prevents the presence of dislocations, implying that the same number of filaments that enter a unit area crossection also exit it.  In the case of the hexagonal columnar phase, nonzero splay would allow for deviations from the lattice structure. 
Heuristically, the constraint \eqref{constraint} is achieved at the limit $K_1\to \infty$ also represented as $K_1 \gg K_2, K_3$.
In the later section where we reconstruct the knotted DNA trajectory, we will relax the constraint by including the splay term corresponding to the elasticity constant.

We recall that the pitch of a helix is the height along the helical axis between two points of the curve separated by a complete $2\pi$ rotation angle. That is,
\begin{equation}
    p(r,\theta,z) = 2 \pi r \tan \psi,  \,\, \,  \psi\in (0, \frac{\pi}{2}).
\end{equation}
As the DNA filament spools and fills the capsid, its cross sections form a hexagonal structure \cite{hud2001cryoelectron,kindt2001,ortiz2003}. We assume that the distance between DNA  layersis of the same order of magnitude as the helical pitch, thus the local DNA density satisfies
\begin{equation}
    \frac{\sqrt{3}}{2} p^2 \times \frac{\rho}{\eta}=1 \Longrightarrow \quad  \rho= \frac{2}{\sqrt 3} \frac{\eta}{p^2} \triangleq \frac{C}{r^2 \tan^2 \psi}
\end{equation}
Here $\frac{\sqrt{3}}{2} p^2$ represents the area of a hexagonal unit cell of diameter $p$. The positive factor  $\eta \approx 3{\text{nm}}^{-1} $ reflects  the fact that each base pair of DNA corresponds to a length of $0.34 \text{nm}$ along the DNA axis \cite{watson1953molecular}. The parameter $ C = \frac{\eta}{2\pi^2 \sqrt{3}}$ has the dimension of base pairs per unit length.

With these assumptions, the total energy of the system can be simplified to,
\begin{eqnarray}
 && E_{total}[\psi(r,\theta,z)]  = \int_\Omega \left[ \frac{K_2}{2}( \frac{\sin(2\psi)}{2r} -\psi_r )^2 + \frac{K_3}{2}\left( (\frac{\cos^2 \psi}{r})^2 + \frac{\psi^2_\theta}{r^2 \cos^2 \psi } \right)\right. \nonumber\\
 && \left. + k_B T\left(\log \frac{C}{r^2 \tan^2 \psi} + \mu\right) \right] \frac{C}{r^2 \tan^2 \psi} d\vec{x}  - \int_{\partial \Omega} \frac{K_2+K_4}{2} \frac{C \cos^2 \psi}{r^3 \tan^2 \psi}\vec{e}_r \cdot d\vec{S} \\
 &\geq & \int_\Omega \left[ k_B T\left(\log \frac{C}{r^2 \tan^2 \psi} + \mu\right)+\frac{K_2}{2}( \frac{\sin(2\psi)}{2r} -\psi_r )^2 \right. \nonumber\\
 && \left.+ \frac{K_3}{2} (\frac{\cos^2 \psi}{r})^2 \right] \frac{C}{r^2 \tan^2 \psi} d\vec{x} - \int_{\partial \Omega} \frac{K_2+K_4}{2} \frac{C \cos^2 \psi}{r^3 \tan^2 \psi}\vec{e}_r \cdot d\vec{S} \\
 &=& \int_{0}^H \int_0^{2\pi} \left(\int_{R_1}^{R_2} \left[ k_B T\left(\log \frac{C}{r^2 \tan^2 \psi} + \mu\right) + \frac{K_2}{2}( \frac{\sin(2\psi)}{2r} -\psi_r )^2 \right. \right.\nonumber\\&& \left. \left. + \frac{K_3}{2} (\frac{\cos^2 \psi}{r})^2  \right] \frac{C}{r \tan^2 \psi} dr -  \frac{K_2+K_4}{2} \left. \frac{C \cos^2 \psi}{r^2 \tan^2 \psi}\right|_{R_1}^{R_2} \right) d\theta dz \\&& \triangleq E[\psi(r,\theta,z)]. \label{OF_energy}
\end{eqnarray}
The minimizer of the energy $E[\psi(r,\theta,z)]$ is expected to be  $\psi(r)$ which is independent of $(\theta,z)$. Furthermore, since $\tan \psi$ appears in the denominator, the concentric configuration $\psi(r) \equiv 0$ would cause the energy $E[\psi(r)]$ to diverge. Therefore, the concentric configuration cannot  be a minimizer. We can further simplify the energy to derive the Euler--Lagrange equation,

\begin{eqnarray}
   \frac{E[\psi(r)]}{\pi HC} &=&  \int_{R_1}^{R_2} \left[ 2k_B T\left(\log \frac{C}{r^2 \tan^2 \psi} + \mu\right) +K_2( \frac{\sin(2\psi)}{2r} -\psi_r)^2 \right. \nonumber\\ && \left. + K_3 (\frac{\cos^2 \psi}{r})^2  \right] \frac{1}{r \tan^2 \psi} dr  - (K_2+K_4) \left. \frac{ \cos^2 \psi}{r^2 \tan^2 \psi}\right|_{R_1}^{R_2} \\
    &=&   \int_{R_1}^{R_2} \left[K_2(\frac{\sin 2\psi}{2r})^2 + K_2 \psi_r^2 + K_3 \frac{\cos^4 \psi}{r^2}  \right. \nonumber\\
    && \left. + 2k_B T\left(\log \frac{C}{r^2 \tan^2 \psi} + \mu\right) \right] \frac{1}{r\tan^2 \psi} dr  \nonumber\\
    && -  \int_{R_1}^{R_2} K_2  \frac{\sin(2\psi)}{r}  \frac{\psi_r}{r \tan^2 \psi} dr - (K_2+K_4) \left. \frac{ \cos^2 \psi}{r^2 \tan^2 \psi}\right|_{R_1}^{R_2}\nonumber\\
    &\triangleq& I_1 + I_2 + I_{sp}
\end{eqnarray}
We first calculate the second integral $I_2$:
\begin{eqnarray}
    I_2 &=& -  \int_{R_1}^{R_2} K_2  \frac{\sin(2\psi)}{r}  \frac{\psi_r}{r \tan^2 \psi} dr\nonumber \\
    &=& -  K_2 \left.\frac{2\ln \sin \psi - \sin^2 \psi }{r^2}\right|_{R_1}^{R_2} -2  K_2\int_{R_1}^{R_2}\frac{2\ln \sin \psi - \sin^2 \psi }{r^3} dr \nonumber\\
    &\triangleq& K_2 \left.\frac{2u + e^{-2u}}{r^2}\right|_{R_1}^{R_2} +2  K_2\int_{R_1}^{R_2}\frac{2u + e^{-2u}}{r^3} dr,
\end{eqnarray}
where we used the change of variables $u = -\ln \sin \psi$. Likewise,
\begin{eqnarray}
    I_1 &=&\int_{R_1}^{R_2} \left[K_2\left( \frac{(1- e^{-2u})^2}{r^3}+  \frac{u_r^2}{r}\right) + K_3 \frac{(1- e^{-2u})^3}{r^3e^{-2u}} \right. \nonumber\\&& \left. + 2 k_B T\left(\log \frac{C(e^{2u}-1)}{r^2} + \mu\right)\frac{e^{2u}-1}{r} \right]  dr.
\end{eqnarray}
The first variation of the functional $E[\psi(r)]$ (equivalently, that of $E[u(r)]$), yields the Euler--Lagrange equation,
\begin{eqnarray}
     &&K_2 \left(\frac{u_r}{r}\right)_r -2\frac{K_3-K_2}{r^3}e^{-4u} - \frac{K_3}{r^3}( e^{2u}-3 e^{-2u}) -  K_2 \frac{2}{r^3}\nonumber\\&=&\frac{4k_B Te^{2u}}{r}\left(\log \frac{C(e^{2u}-1)}{r^2} + \mu + 1 \right).   \label{eqn-u}
\end{eqnarray}
The natural boundary conditions stemming from the first variation of the total energy are
\begin{equation}
\begin{cases}
    K_2 R_1 u_r(R_1) = -K_2 (1 - e^{-2u(R_1)}) + (K_2 + K_4) (e^{2u(R_1)}-e^{-2u(R_1)}),\\
    K_2 R_2 u_r(R_2) = -K_2 (1 - e^{-2u(R_2)}) + (K_2 + K_4) (e^{2u(R_2)}-e^{-2u(R_2)}).
    \end{cases}
\end{equation}
Alternatively, we impose Dirichlet boundary conditions,
\begin{equation}
    u(R_1) = M_1, \ \ u(R_2) = M_2. \label{dirichlet_BC}
\end{equation}
In the following sections, we impose  Dirichlet boundary conditions on the helical angle $\psi(r)$, guided by physical considerations at  the inner and outer radii of the ordered DNA region. 

At the outer boundary $r = R_2$, corresponding to the capsid wall, we assume a vanishing helical pitch $p = \epsilon \to 0$,  which reflects a strong anchoring to a concentric  spooling configuration. Accordingly, the  helical angle takes the form  $\psi(R_2) = \arctan\left(\frac{\epsilon}{2\pi R_2}\right)$, and is small in the limit $\epsilon\to 0.$ Consequently, the boundary value $M_2 = \left| \ln \sin \psi(R_2) \right|$ is large, since $\sin \psi(R_2)$ is close to zero.

At the inner boundary $r = R_1$, where the disordered core begins, we assume that the helical pitch equals the radial position $R_1$. This reflects the assumption that only a single DNA layer can be accommodated within the disordered region if the inter-strand spacing exceeds $R_1$. Imposing the condition $2\pi R_1 \tan \psi(R_1) = R_1$  yields the helical angle 
\begin{equation}
    \psi(R_1) = \arctan\left(\frac{1}{2\pi}\right),\label{threshold_angle}
\end{equation}
and the corresponding boundary value  $M_1 = \left| \ln \sin \psi(R_1) \right|$. These boundary conditions ensure that the director field transitions smoothly from the tightly packed configuration near the capsid wall to the more open structure near the disordered core.

\subsection{Existence and uniqueness of solutions} In order to show that the boundary value problem consisting of  Eq. \eqref{eqn-u} and \eqref{dirichlet_BC} has a unique solution, we first rewrite it in terms of the new independent variable, 
\begin{equation}
    y = \log \frac{r}{R_1},\ \ \  y \in [0, \log R_2/R_1].
\end{equation}
The Euler--Lagrange equation becomes,
\begin{eqnarray}
    &&\big(e^{-2y}u_y\big)_y - 2e^{-2y}\left[(\alpha-1)e^{-4u} + \frac{\alpha}{2}(e^{2u}-3 e^{-2u}) + 1\right]\nonumber\\
    &=&\omega^2 \left( \log \frac{C(e^{2u}-1)}{R_1^2} -2y + \mu + 1 \right)e^{2u}, \quad y\in [0, \log\frac{R_2}{R_1}], \quad \text {with}\nonumber\\ && \quad \quad \alpha = K_3/K_2\,\, \text{ and }\,\, \omega^2= \frac{4R_1^2 k_B T}{K_2}.
    \label{hy-eqn}
\end{eqnarray}
The Dirichlet boundary conditions give
\begin{equation}
    u(0) = M_1, \ \ u(\log\frac{R_2}{R_1}) = M_2. \label{hy-BC1R}
\end{equation}

Let us introduce the notation
\begin{align}
    & G(u)= (\alpha-1) e^{-4u} +\frac{\alpha}{2}(e^{2u}-3 e^{-2u})+1,  \label{G}\\
    & h(y, u(y))= 2e^{-2y}G(u(y)),  \label{h}\\
    &  g(y, u(y)) =   \omega^2\left( \log \frac{C(e^{2u}-1)}{R_1^2} -2y + \mu + 1 \right)e^{2u}. \label{g}
\end{align}
The boundary value problem \eqref{hy-eqn}, \eqref{hy-BC1R} now takes the  form,
\begin{align}
(e^{-2y}&u_y)_y=f(y, u(y)) \triangleq h(y,u(y)) + g(y,u(y)), \quad y\in I:=[0, \, \log\frac{R_2}{R_1}]\label{uy} \\
&u(0)=M_1 < u(\log \frac{R_2}{R_1})=M_2. \label{uy-BCR1}
\end{align}
Note that $G'(u)= 4(1-\alpha) e^{-4u} +\alpha(e^{2u}+3 e^{-2u})>0$, so the function $G(u)$ is monotonic increasing on $u \in [0,\infty)$ with $G(0) =0$ and $\displaystyle \lim_{u\to +\infty} G(u) = + \infty$. Similarly, with fixed $y \in [0, \log \frac{R_2}{R_1}]$, it follows that $\displaystyle \lim_{u \to 0}  g(y,u) = -\infty$ and $\displaystyle \lim_{u \to +\infty}  g(y,u) = +\infty$. With these properties, we can estimate the maximum and minimum of the solution $u(y)$.
\begin{lemma}
Let $u(y) \geq 0$ be a  classical solution to the boundary value problem \eqref{uy}, \eqref{uy-BCR1}. Then $u(y)<u_{max}=max(M_2,\frac{1}{2} \log \left( \frac{R_2^2}{C} e^{-\mu - 1 } + 1\right))$, for all $y\in I$.
\end{lemma}
\begin{proof}
    If the maximum of $u(y)$ is attained on the boundary, then $max(u(y))=M_2$. Otherwise, the maximum of $u(y)$ is attained at $y = y_0 \in (0,\log\frac{R_2}{R_1})$, then $u_y(y_0)=0$. Suppose $u(y_0)>\frac{1}{2} \log \left( \frac{R_2^2}{C} e^{-\mu - 1 } + 1\right)$, then $ g(y_0,u(y_0)) > 0$. Using Eq. \eqref{uy}, we have $u_{yy}(y_0) > 0$, which means $u(y_0)$ is a local minimum, contradicts with $u(y_0)$ is the maximum of $u(y)$.
\end{proof}

\begin{lemma}
    Let $u(y) \geq 0$ be a classical solution to the boundary value problem \eqref{uy}, \eqref{uy-BCR1}. Then there exists $u_{min}>0$, such that $u(y)>u_{min}$, for all $y\in I.$
\end{lemma}
\begin{proof}
    If the minimum of $u(y)$ is attained on the boundary, then $min_{y\in I}u(y)=M_1$. Otherwise, suppose that the minimum of $u(y)$ occurs at an interior point $y = y_0 \in (0,\log\frac{R_2}{R_1})$. In this case, we have $u_y(y_0)=0$, and we consider the behavior of the second derivative at $y_0$.  From the properties of the functions $G(u)$ and $g(y,u)$, there exists a value $u_{min}$ (possibly depending on $\mu$ and $\alpha$) such that $f(y,u_{min}) < 0$,  for all $y \in (0,\log\frac{R_2}{R_1})$.  Now, suppose for contradiction that  $u(y_0) < u_{min}$. Then, evaluating the differential equation at $y_0$,  we find that  $u_{yy}(y_0)<0$, implying that  $u(y_0)$ is a local maximum of $u $.  This contradicts the assumption that that $u(y_0)$ is the minimum of $u(y)$.
    Therefore, it must be that $u(y)\geq u_{min}  $ for all $y\in I$, and the minimum of $u$ is either equal to $M_1$ (if on the boundary) or greater or equal than $u_{min}$ (if in the interior). 
\end{proof}

 We can now establish the following results for existence and uniqueness of solution.
 \begin{theorem} \label{SL}
 Let $0<R_1<R_2$ and $0<M_1<M_2$. The boundary value problem for equations \eqref{hy-eqn}-\eqref{hy-BC1R} has a unique  classical solution.
 \end{theorem}

\begin{proof} First, we apply a change of variables to transform the boundary conditions into homogeneous ones. Specifically, we define
\begin{equation}
    \tilde{u}(y) = u(y) - M_1 - \frac{M_2 - M_1}{\log \frac{R_2}{R_1}} y.
\end{equation}
So, the boundary value problem can be written as,
\begin{equation}
    (e^{-2y} \tilde{u}_y)_y = f(y, u) + 2 \frac{M_2 - M_1}{\log \frac{R_2}{R_1}}, \ \ \tilde{u}(0) = 0, \ \ \tilde{u}(\log \frac{R_2}{R_1}) = 0.
\end{equation}
Its solution is bounded by $\tilde{u}_{min} = u_{min} - M_2$, and $\tilde{u}_{max} = u_{max} - M_1$, with $u_{min}$ and $u_{max}$ from Lemma 1 and Lemma 2. Now define $A(y) = e^{-2y}$, and
\begin{equation}
    \tilde{f}(y,\tilde{u},v) = \begin{cases}
        f(y,u_{min}) + 2 \frac{M_2 - M_1}{\log \frac{R_2}{R_1}}, \ \ \ \text{ if }\  \tilde{u}+ M_1 + \frac{M_2 - M_1}{\log \frac{R_2}{R_1}} y<u_{min},\\
        f(y,u_{max})+ 2 \frac{M_2 - M_1}{\log \frac{R_2}{R_1}},  \ \ \ \text{ if }\  \tilde{u}+ M_1 + \frac{M_2 - M_1}{\log \frac{R_2}{R_1}} y>u_{max},\\
        f(y,\tilde{u}+ M_1 + \frac{M_2 - M_1}{\log \frac{R_2}{R_1}} y)+ 2 \frac{M_2 - M_1}{\log \frac{R_2}{R_1}}, \hspace{40pt} \text{otherwise.}
    \end{cases}
\end{equation}
and consider the boundary value problem,
\begin{equation}
    (A(y) \tilde{u}_y)_y = \tilde{f}(y, \tilde{u}, \tilde{u}_y),  \ \ \tilde{u}(0) = 0, \ \ \tilde{u}(\log \frac{R_2}{R_1}) = 0,
\end{equation}
We can readily verify, the function $\tilde{f}(y,\tilde{u},v)$ satisfies the Carath\'eodory conditions:
    \begin{enumerate}
    \item For every $y\in [0,\log\frac{R_2}{R_1}]$, the mapping  $ (\tilde{u},v) \longrightarrow  \tilde{f}(y, \tilde{u}, v)$ is continuous.
    \item For every $ (\tilde{u},v) \in \mathbb R^2, $ the mapping $y\longrightarrow \tilde{f}(y, \tilde{u}, v)$ is measurable.
    \item For every $r>0$, there exits a function $g(y)\in L^1([0,\log\frac{R_2}{R_1}])$ such that, for all $|\tilde{u}|\leq r$, $|v|<r$, and $y\in [0,\log\frac{R_2}{R_1}]$, 
    $$ \displaystyle |\tilde{f}(y, \tilde{u}, v)| \leq g(y) \triangleq \max_{u \in [u_{min},u_{max}]} f(y,u)+ 2 \frac{M_2 - M_1}{\log \frac{R_2}{R_1}}.$$
\end{enumerate}
Moreover, the following properties hold:
\begin{enumerate}
    \item There exists a positive constant $\nu = \frac{R_1^2}{R_2^2}$ such that $(\xi, A(y) \xi) \geq  \nu |\xi|^2$, for all $\xi \in \mathbb R$ and $y \in I$.
    \item We have that  $\int \frac{1}{A(y)} dy = \frac{1}{2}e^{2y}>0$.
    \item For every $(\tilde{u},v) \in \mathbb R^2$, we have \,\, 
     $   f(y,\tilde{u},v) \leq g(y).$
\end{enumerate}
Therefore, by Theorems 3.1 and 4.1 in \cite{chern1995nonlinear}, there exists a unique solution $\tilde{u}(y)$. Consequently,  there also exists a unique  classical solution $u(y)$ to the boundary value problem \eqref{uy} - \eqref{uy-BCR1}.
\end{proof}

\section{DNA Reconstruction with Randomization}

The DNA trajectory can be reconstructed from its tangent vector field $\vec{n}$ by solving the  initial value problem,
\begin{equation}
    \frac{d \vec{r}}{d\xi} = \vec{n}(\vec{r}(\xi)), \,\vec{r}(0) =\vec{r}_0,
\end{equation}
where $\xi$ denotes the arc length along the DNA curve and $\vec{r}(0)$ specifies the position of one end of the reconstructed DNA molecule.  The tangent field $\vec{n}(\vec{r}(\xi))$ is taken as  the  minimizer of the Oseen--Frank theory. However, such a reconstruction yields a perfect helical curve without  knots, since the Oseen--Frank model is a mean-field approximation, that only captures the average molecular orientation and does not explicitly account for thermal fluctuations.

To overcome this limitation, we incorporate local thermal fluctuations into the DNA curve reconstruction using a Langevin approach, in which stochastic noise is introduced at the force level, balancing the restoring forces arising from the first variation of the Oseen--Frank energy, and the random force from Brownian motion. 
In cylindrical coordinates, the randomized DNA trajectory is represented as 
\begin{equation}
    \vec{r}(\xi) = (r(\xi), \theta(\xi), z(\xi)), \ \ \ \ \ \begin{cases}
        \displaystyle \frac{dr}{d\xi} = \sin \beta,\\
        \displaystyle \frac{d\theta}{d\xi} = \frac{\cos \beta \cos{\psi}}{r},\\
        \displaystyle \frac{dz}{d\xi} = \cos \beta \sin{\psi}.
    \end{cases}
    \label{Trajectory_Equations}
\end{equation}
The tangent direction is given by 
\begin{equation}
\vec{n} = (\sin \beta, \cos \beta \cos \psi, \cos \beta \sin \psi)
\end{equation}
which represents the perturbation of the average vector field $\vec{n}_0 = (0, \cos \psi_0, \sin \psi_0)$ obtained from the Oseen--Frank theory as the solution to Eq. \eqref{uy}. The stochastic angles $\beta(\xi)$ and $\psi(\xi)$,  describing deviations from the mean orientation, evolve according to the Langevin's equation, which is also known as Ornstein--Uhlenbeck process,
\begin{eqnarray}
\begin{cases}
    d\beta = \sigma_\beta dB_\beta - \kappa_\beta  \beta d{\xi},\\
    d\psi = \sigma_\psi dB_\psi - \kappa_\psi (\psi - \psi_0) d{\xi}
\end{cases}\label{LangevinNoise}
\end{eqnarray}
with initial condition $\beta(0) = 0$ and $\psi(0) = \psi_0$. Here $dB_\beta$ and $dB_\psi$ denote independent  Brownian motions and $\sigma_\beta$ and $\sigma_\psi$ are constants quantifying the thermal fluctuation.
The terms $\kappa_\beta \beta$ and $\kappa_\psi (\psi - \psi_0)$ represent restoring forces that drive the vector $\vec{n}$ towards the mean orientation $\vec{n}_0$.  The competition  between the two forces in Eq. \eqref{LangevinNoise} causes the angles to deviate from the mean-field description, while  the restoring forces control these deviations, leading to a Gaussian-type distributions of the angles.

Before estimating the constants $\kappa_\beta, \kappa_{\psi}$ and $\sigma_{\beta}, \sigma_{\psi}$, we first review the properties of the proposed system, in the linear response region. When all parameters $\sigma$ and $\kappa$, and $\psi_0$ are fixed constants, the Ornstein-–Uhlenbeck process \cite{kloeden1992stochastic} implies that  $\beta({\xi})$ and $\psi({\xi})$ are independent Gaussian random variables, with distributions given by
\begin{equation}
\begin{cases}
     \beta(\xi) \sim \mathcal{N}[0, \frac{\sigma_\beta^2}{2\kappa_\beta}(1-e^{-2\kappa_\beta {\xi}})],\\
    \psi(\xi) \sim \mathcal{N}[\psi_0, \frac{\sigma_\psi^2}{2\kappa_\psi}(1-e^{-2\kappa_\psi {\xi}})].
\end{cases}
\end{equation}
Here $\mathcal{N}(\mu,\sigma^2)$ represents the normal distribution with mean $\mu$ and variance $\sigma^2$. The expectations are consistent with the mean-field solution in Section \ref{continuum_model}, and the variance proportional to the ratio $\frac{\sigma^2}{\kappa}$.

The expectation of the tangent vector is given by
\begin{equation}
    \mathbb{E}[\vec{n}(\xi)] = \exp\left(-\frac{\sigma_\beta^2}{4\kappa_\beta}(1-e^{-2\kappa_\beta \xi}) -\frac{\sigma_\psi^2}{4\kappa_\psi}(1-e^{-2\kappa_\psi\xi})\right) (0, \cos \psi_0 , \sin \psi_0  ).
\end{equation}
This formulation  ensures that the randomized DNA curve, on average,  follows the same tangent direction as  the Oseen--Frank solution. Moreover, the variances of $\beta(\xi)$ and $\psi(\xi)$ remain bounded as $\xi \to \infty$, indicating  that deviations from the mean configuration are limited.  Assuming that $\beta$ is small and linearizing Eq. \eqref{Trajectory_Equations}, we obtain
\begin{equation}
\begin{cases}
    \mathbb{E}[r(\xi)] = r_0,\ \ \  \text{Var}[r(\xi)] = \frac{\sigma_\beta^2}{2\kappa_\beta}\left(\xi - \frac{1-e^{-2\kappa_\beta \xi}}{2\kappa_\beta}\right),\\
    \mathbb{E}[z(\xi)] = z_0 + \xi \sin \psi_0,\ \ \  \text{Var}[z(\xi)] = \frac{\sigma_\psi^2}{2\kappa_\psi}\left(\xi - \frac{1-e^{-2\kappa_\psi }\xi}{2\kappa_\beta}\right),
    \end{cases}\label{rz-variance}
\end{equation}
representing the accumulation of the mean and variance.

For long reconstructed curves, the  variance grows with $\xi$,   
causing individual realizations of 
${r}(\xi)$ to deviate significantly from the mean value ${r}_0$.
Since  the DNA trajectory must remain confined within the capsid, in reconstructing the trajectory    we need to further impose the condition $r(\xi) \leq R_2$. This   prevents $r(\xi)$ from becoming unphysically  large. It is achieved by modifying the first equation in \eqref{Trajectory_Equations} to the following:  
\begin{equation}
    \frac{dr}{d\xi} = \begin{cases}
        \sin \beta, \text{ if } r(\xi)<R_2 \text{ and } \beta \leq 0,\\
        0, \text{ otherwise.}
    \end{cases} \label{radius_correction}
\end{equation}
The situation near  the inner radius  $R_1$ is different, since there is no physical boundary 
preventing DNA from crossing into the disordered core. Instead of  enforcing $r > R_1$,  we regulate the dynamics by appropriately choosing $\kappa$ and $\sigma$, as described in the next section, so that   $r(\xi)$ does not become too small. 

\subsection{Estimation of the restoring forces}

The terms $\kappa_\beta \beta$ and $\kappa_\psi (\psi - \psi_0)$ in Eq. \eqref{LangevinNoise} describe the linear elastic, restoring forces  that act on the system, when the unit vector field $\vec{n}$ deviates from $\vec{n}_0$. To estimate the coefficients $\kappa_\beta$ and $\kappa_\psi$, we consider the perturbations separately, while ignoring the crossing terms, and focus on the contribution from the Oseen--Frank elastic energy. %

\paragraph{Perturbation in $\psi$} Consider the perturbed vector $\vec{n} = (0, \cos \psi, \sin \psi)$, with $\beta = 0$ and $\psi = \psi_0 + \delta \psi$. So the perturbed energy,
\begin{equation}
    E + \delta E = \int_\Omega \frac{K_3 |\vec{n} \times \nabla \times \vec{n}|^2 + K_2 (\vec{n} \cdot \nabla \times \vec{n})^2 + K_1 |\nabla \cdot \vec{n}|^2}{r^2 \tan^2(\psi_0+\delta \psi)} d\vec{x}
\end{equation}
The forces are given by the first variation of the energy $\frac{\delta E}{\delta \psi}(\psi)$. Since $\psi_0$ is the minimizer of the energy, so the Euler--Lagrange equation \eqref{uy} is equivalent to $\frac{\delta E}{\delta \psi}(\psi_0) = 0$. Then the linear force induced by a small perturbation around $\psi_0$ could be given by the Taylor's expansion of $\frac{\delta E}{\delta \psi}(\psi_0+\delta \psi)$. Moreover, since we are only interested in the coefficient of $\delta \psi$, we ignore the terms that come from $\delta \psi_r$, $\delta \psi_\theta$ and $\delta \psi_z$, 
thus,
\begin{equation}
    \frac{\delta E}{\delta \psi}\left(\psi_0 + \delta \psi \right) \propto \frac{1}{\tan^2 \psi_0} \left[\frac{K_2}{r^3}8e^{-4u} + \frac{K_3}{r^3}(2e^{2u}+6e^{-2u}-8e^{-4u})\right] \delta \psi. \label{force_psi}
\end{equation}
Notice, the function $\psi_0$ is undefined if $r < R_1$, thus $\psi_0 $ will be treated as a constant for each layer of reconstruction.
\paragraph{Perturbation in $\beta$} Following a similar approach, we consider the perturbed vector field, $\vec{n} = (\sin \beta, \cos \beta \cos \psi, \cos \beta \sin \psi)$.
Since $K_1 \gg K_2, K_3$, the $K_1$ term will dominate in the perturbed energy,
\begin{equation}
    \frac{\delta E}{\delta \beta}(\beta) \propto \frac{\beta}{r^3 \tan^2 \psi} K_1. \label{force_beta}
\end{equation}
For simplicity, we neglect the dependence on $\psi_0$ or $u$, as they are only defined on $[R_1,R_2]$, so that we estimate 
\begin{equation}
\kappa_\beta(r) = \frac{C_1 R_1^2}{r^3},\ \ \kappa_\psi(r) = \frac{C_2 R_1^2}{r^3}, \ \ \text{ for } \ r \in (0,R_2]. \label{kappa_r_dependent}
\end{equation}
Here $R_1^2$ is a constant chosen such that the constant coefficients $C_1$ and $C_2$ being dimensionless; their values will be determined in the numerical section. Comparing the forces in Eq. \eqref{force_psi} and \eqref{force_beta}, $C_1 \gg C_2$ as $K_1 \gg K_2$ and $K_3$. The form given by Eq. \eqref{kappa_r_dependent}, indicates the restoring force being strong when radius is small, thus preventing the reconstructed DNA filament from being  close to the center of the capsid.

\subsection{Estimation of noise terms} We  expect that different layers of helical curves may cross at certain locations (for  relatively small arc length $\xi$), hence allowing the formation of knots. To capture this effect,  we impose that the standard deviation of the radial displacement and the height displacement are comparable to the spacing between adjacent  DNA strands. This leads to the condition
\begin{equation}
    \sqrt{2 \pi r \frac{\sigma_\beta^2}{\kappa_\beta} } = \sqrt{2 \pi r \frac{\sigma_\psi^2}{\kappa_\psi} }\sim 2 \pi r \tan \psi_0.
\end{equation}
Assuming that, both $\sigma_\psi$ and $\sigma_\beta$ are constants representing the thermal fluctuations, we  express them as,
\begin{equation}
 \sigma_\beta = \sqrt{C_1} C_3 \tan \psi_0, \ \ \ \sigma_\psi = \sqrt{C_2} C_3 \tan \psi_0.
\end{equation}
Here $\psi_0=\psi(r_0)$ is treated as a constant during the reconstruction. The expression $C_3$ is thus a dimensionless constant to  be determined in the numerical section. 


\section{Numerical Methods and Examples} In this section, we present the numerical solutions obtained for the DNA curve reconstruction.

\subsection{Solution of the Euler--Lagrange Equation} \label{numerical_solution_pde} We first consider the solutions to Eq. \eqref{uy}-\eqref{uy-BCR1} under different parameter values of $\alpha$, $\mu$ and $\omega^2$.  The inner radius is fixed  at $R_1 = 10 nm$, consistent with estimates  for bacteriophage P4 \cite{walker2020fine}, while the outer radius is taken as  $R_2 = 20 nm$ \cite{shore1978determination}.

The system is solved using the finite element method implemented in Firedrake \cite{FiredrakeUserManual}, based on its weak formulation,
\begin{equation}
    \int_0^{\log \frac{R_2}{R_1}} e^{-2y} u_y v_y + f(y,u) v dy =0.
\end{equation}
The nonlinear terms in $f(y,u)$ are treated  using the built-in Newton iteration scheme.

\begin{figure}[htbp]
    \centering
    \includegraphics[width=0.32\linewidth]{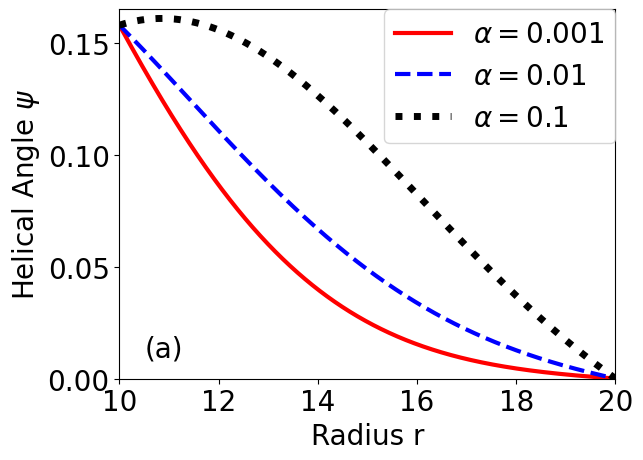}
    \includegraphics[width=0.32\linewidth]{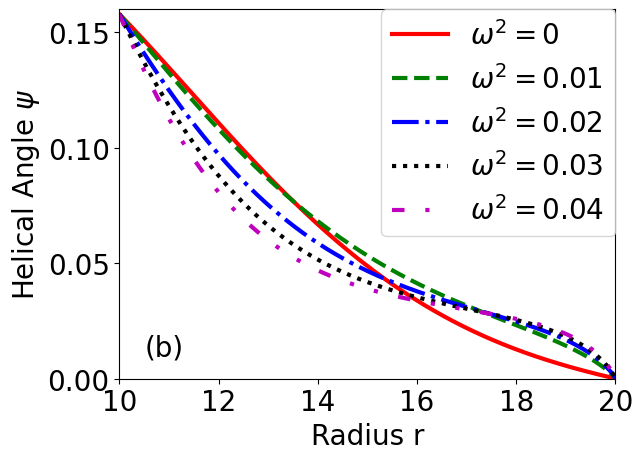}
    \includegraphics[width=0.32\linewidth]{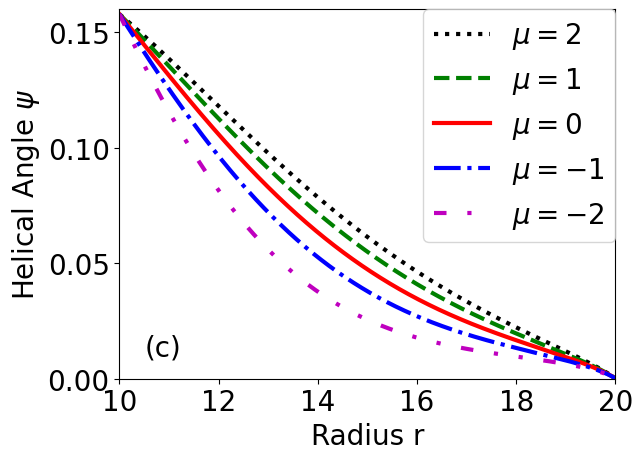}
    \includegraphics[width=0.32\linewidth]{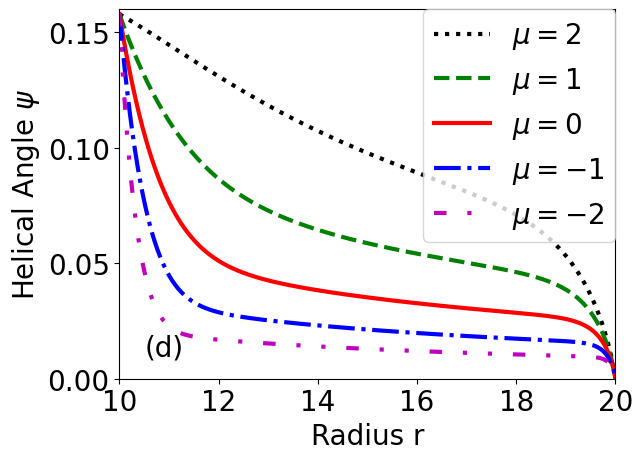}
    \includegraphics[width=0.32\linewidth]{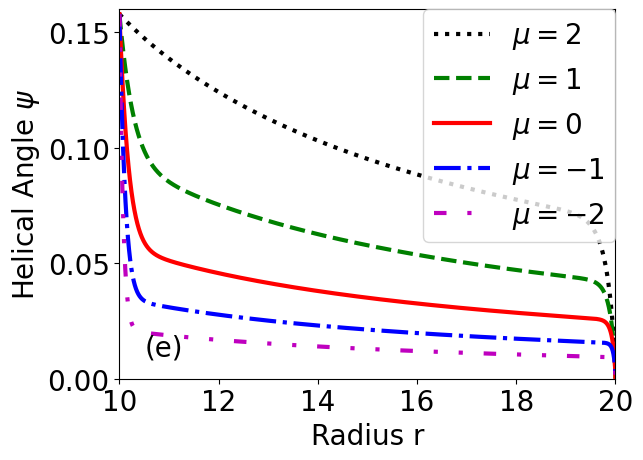}
    \includegraphics[width=0.32\linewidth]{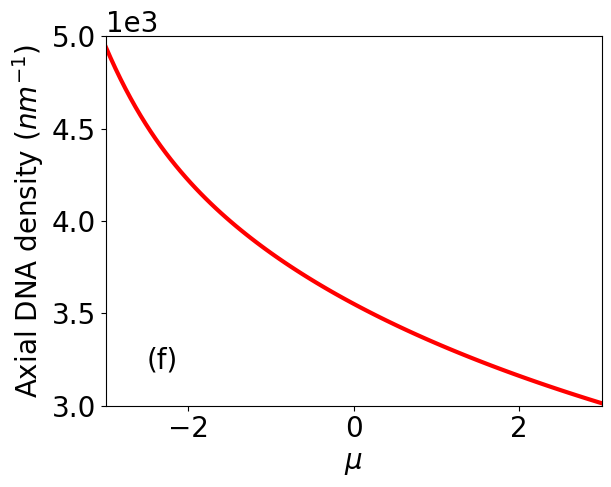}
    \caption{Helical angle $\psi(r)$ on $[R_1,R_2]$, given by Eq. \eqref{uy}. Panel (a): Simplified Eq. \eqref{uy_simplify}. Panel (b): $\alpha = 0.01$. Panel (c): $\alpha = 0.001$, $\omega^2 = 0.003$. Panel (d): $\alpha = 0.1$, $\omega^2=0.3$. Panel (e): $\alpha = 2$, $\omega^2 = 6$. Panel (f): Packed DNA length, $\alpha = 0.001$, $\omega^2 = 0.003$. The axial density is uniform in $z$ direction. }
    \label{solution_PDE}
\end{figure}
Figure \ref{solution_PDE} (a) shows the case when $\omega^2=0$ in which Eq. \eqref{uy} simplifies to 
    \begin{equation}
    \big(e^{-2y}u_y\big)_y - 2e^{-2y}\big((\alpha-1)e^{-4u} + \frac{\alpha}{2}(e^{2u}-3 e^{-2u}) + 1\big) = 0.\label{uy_simplify}
\end{equation}
Following a similar argument as in Theorem \ref{SL}, the equation with Dirichlet boundary conditions Eq. \eqref{uy-BCR1} admits a unique solution. We note that Eq. \eqref{uy_simplify} corresponds to the regime in which the Oseen--Frank energy dominates over the entropic and chemical potential contributions, so that we neglect the entropic and chemical potential terms. In this setting, the Frank constants appearing in Eq.~\eqref{uy_simplify} should be interpreted as effective parameters that incorporate additional influences, such as the finite confinement and other microscopic interactions. 

The validity and limitations of such a simplified model are shown as in Panel (a), where the helical angle $\psi(r)$ depends solely on $\alpha$:  smaller values of  $\alpha$ produce smaller helical angles near  the capsid radius and higher angles near the disordered region. For example,  when $\alpha = 0.1$, the helical angle $\psi$ becomes non-monotonic and  may exceed   the threshold angle given by Eq. \eqref{threshold_angle}. Comparing with other panels, the shape of $\psi(r)$ is well captured, except at the two ends near $R_1$ and $R_2$.

Returning to the general equations \eqref{uy}-\eqref{uy-BCR1}, 
we estimate the value of $\omega^2$ using the expression for the bending modulus \cite{tzlil2003,klug2003director}:
\begin{equation}\rho K_3 = k_B T l_p \rho/\eta.
\end{equation} 
Assuming a DNA persistence length of $l_p = 50 nm$, the definition of $\omega^2$ in Eq. \eqref{hy-eqn} reduces to: 
\begin{equation}
   \omega^2 = \frac{4 R_1^2 k_B T}{K_2} = \frac{4 R_1^2 K_3 \rho}{K_2 l_p \rho/\eta} = \frac{4 (10 nm)^2 K_3}{K_2 50 nm \times 3 nm} \approx 3 \alpha. \label{Bending_modulus}
\end{equation}
 We therefore conclude that $\omega^2$ and $\alpha$ are of the same order of magnitude and note that the actual proportionality factor  depends on the value of $R_1$, which for $P4$ phage is currently an estimate.

Panel(b) in Figure \ref{solution_PDE} shows the solution for $\alpha = 0.01$, with $\omega^2$ ranging from $0$ to $0.04$. We see that the difference in the helical angle remain negligible for all values of $\omega^2$ near the capsid $R_2=20$. However, these differences gradually increase with decreasing radius, reaching  up to 0.05 for $r=13$.

Panels (c), (d), (e) in Figure \ref{solution_PDE} illustrate  the role of the chemical potential of DNA, $\mu$. Across these panels, we observe that as $\mu$ decreases, the helical angle becomes smaller, with  the changes in $\psi$ being relatively small throughout,  except near $R_1$ and $R_2$.  These relations imply that as $\mu$ decreases, more DNA is packed into the capsid. The dependence of the  packed DNA length on $\mu$ is shown in Panel (f) in Figure \ref{solution_PDE}.

\subsection{DNA Reconstruction}\label{reconstruction_one_layer} We consider the reconstruction of the randomized DNA curve given by Eq. \eqref{Trajectory_Equations}. As discussed above, we take the values for bacteriophage P4:  $R_1 = 10\ nm$ and $R_2=20\ nm$. The total reconstructed length in each layer is $2000\ nm$, mimicking the fully packaged $4000\ nm$ genome of bacteriophage P4. 

Eq. \eqref{Trajectory_Equations} is solved using the Euler--Maruyama's method, which is equivalent to Milstein's method since both $\sigma_\beta$ and $\sigma_\psi$ are  constants. This numerical scheme has both 1st order weak convergence and strong convergence.

\begin{figure}[!htbp]
    \centering
    \includegraphics[width=0.32\linewidth]{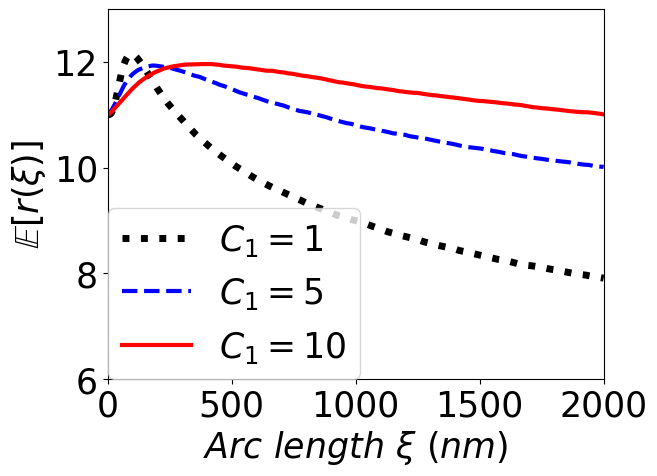}
    \includegraphics[width=0.32\linewidth]{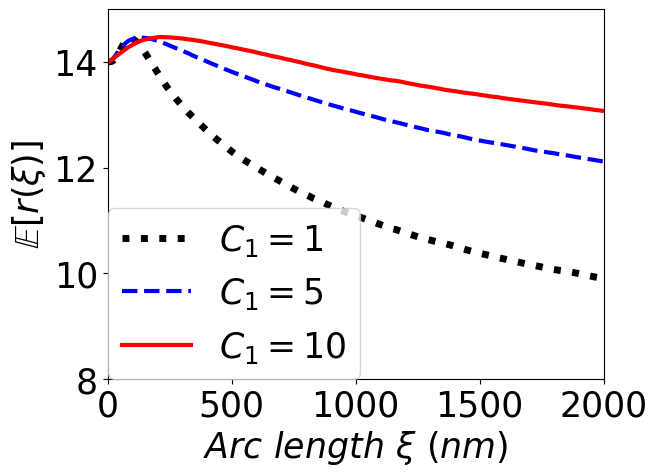}
    \includegraphics[width=0.32\linewidth]{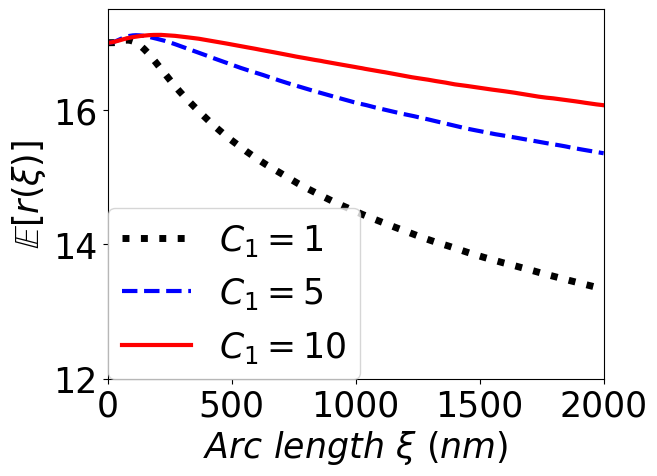}
    \includegraphics[width=0.32\linewidth]{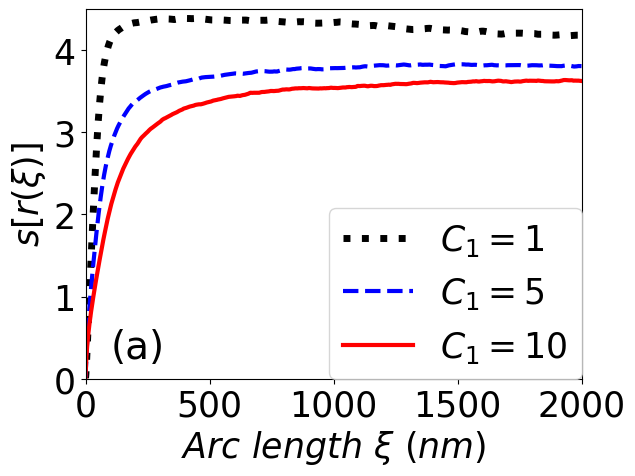}
    \includegraphics[width=0.32\linewidth]{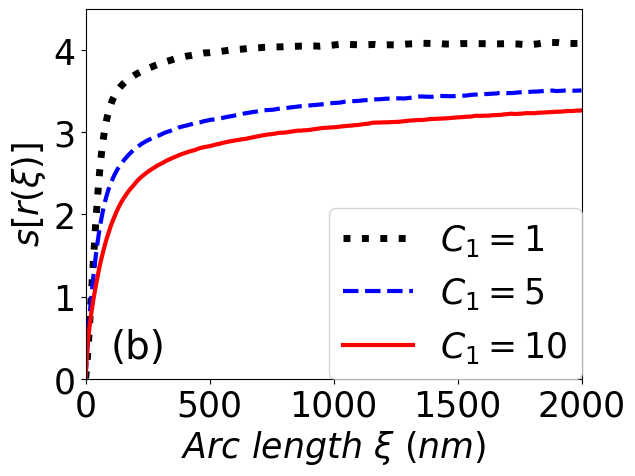}
    \includegraphics[width=0.32\linewidth]{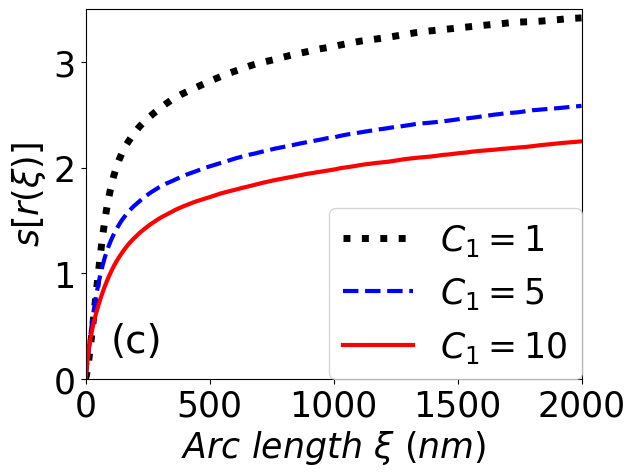}
    \includegraphics[width=0.32\linewidth]{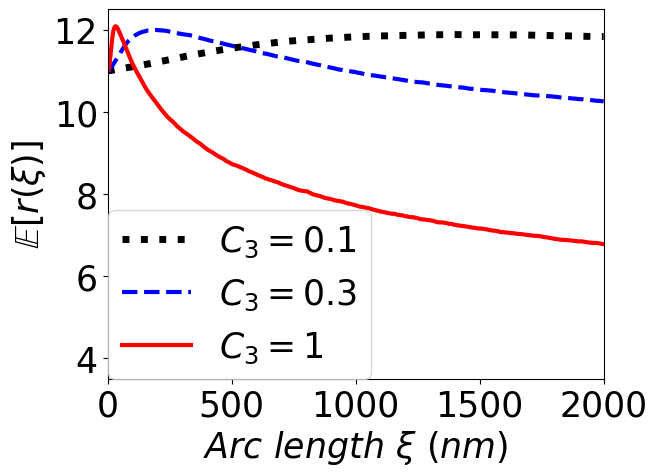}
    \includegraphics[width=0.32\linewidth]{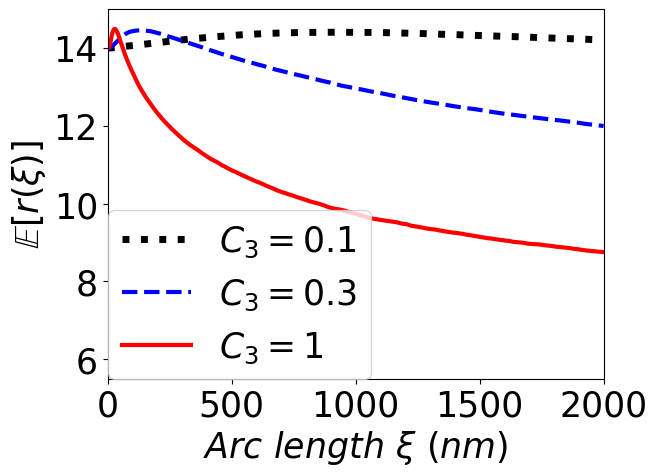}
    \includegraphics[width=0.32\linewidth]{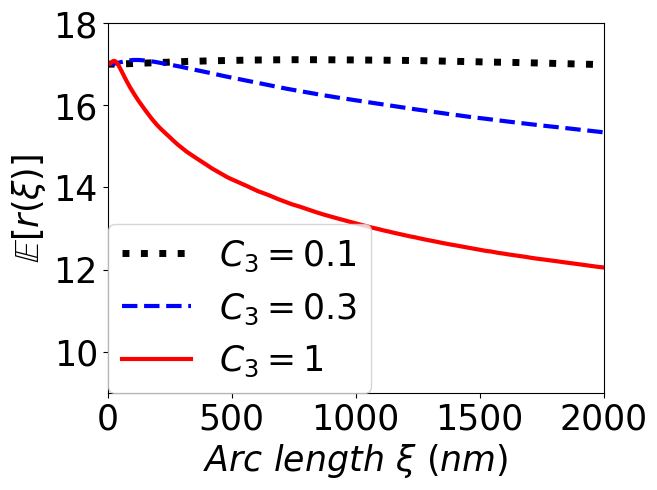}
    \includegraphics[width=0.32\linewidth]{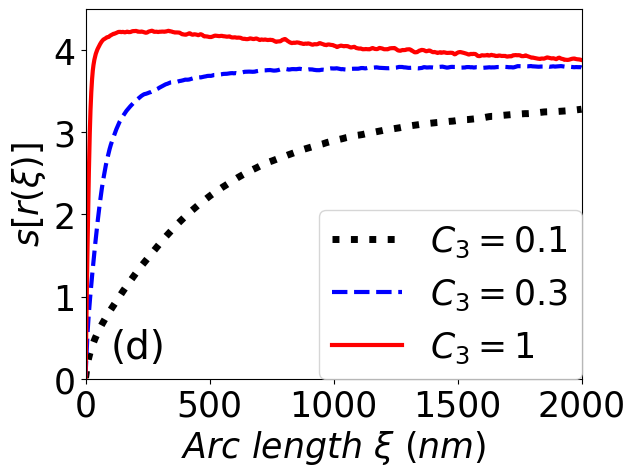}
    \includegraphics[width=0.32\linewidth]{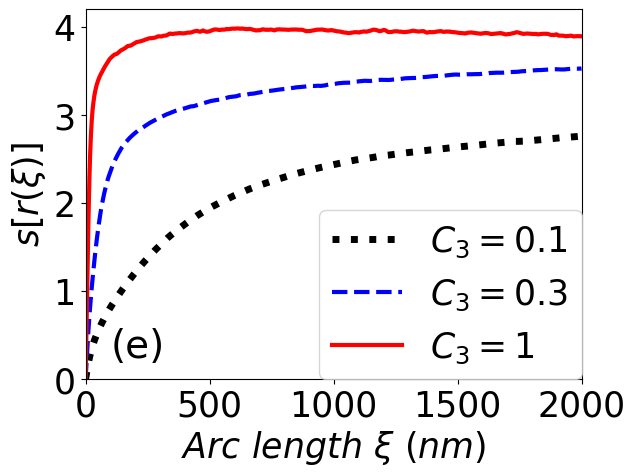}
    \includegraphics[width=0.32\linewidth]{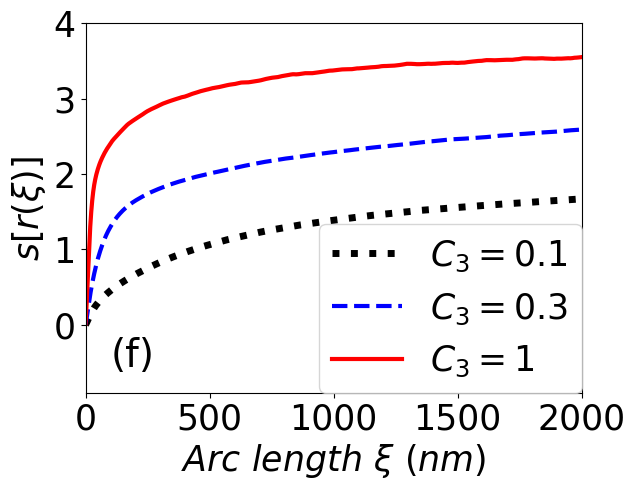}

    \caption{Statistic mean and standard deviation of $r(\xi)$, given by Eq. \eqref{Trajectory_Equations} and \eqref{LangevinNoise}. For each panel, $N=100000$ realizations/paths are generated to evaluate the sample mean and sample standard deviation. Helical angle $\psi_0$ is obtained from the case $\alpha=0.001$, $\omega^2=0$. Panels (a), (b), (c): $C_3=0.3$. Panels (d), (e), (f): $C_1=5$. Panel (a), (d): $r_0 = 11$; Panel (b), (e): $r_0 = 14$; Panel (c), (f): $r_0 = 17$.}
    \label{mean_variance_radius}
\end{figure}

In Fig. \ref{mean_variance_radius}, to properly choose the randomization constants $C_1$ and $C_3$, we plot the expectations and standard deviations of the reconstructed DNA trajectory, with initial radius $r_0=11$, $14$ and $17$ separately. For all of the cases, with a small value of $C_1$, equivalently the restoring force given by $\kappa_\beta$ being weak, the mean value of the DNA curve deviates from the initial $r_0$, as the arc length increases. On the other hand, when $C_1$ is large, the standard deviation of the DNA curve becomes relatively small. Similarly, when $C_3$ is large, representing strong thermal fluctuation, the deviation of the mean value is significant; and when $C_3$ is small, the standard deviation of the reconstructed samples becomes small. In practice, we aim for  the mean value to remain close to $r_0$ preserving a well-defined layering structure in the reconstruction, while allowing the variance to be nonzero, so that intersections between different layers are possible. 
\begin{figure}[!htbp]
    \centering
    \includegraphics[width=0.32\linewidth]{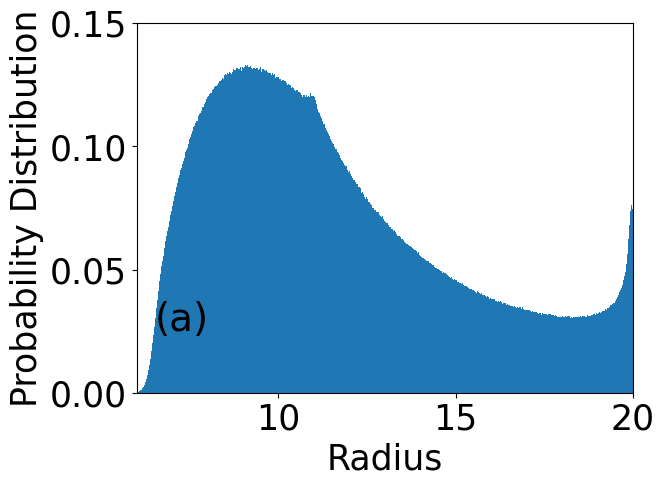}
    \includegraphics[width=0.32\linewidth]{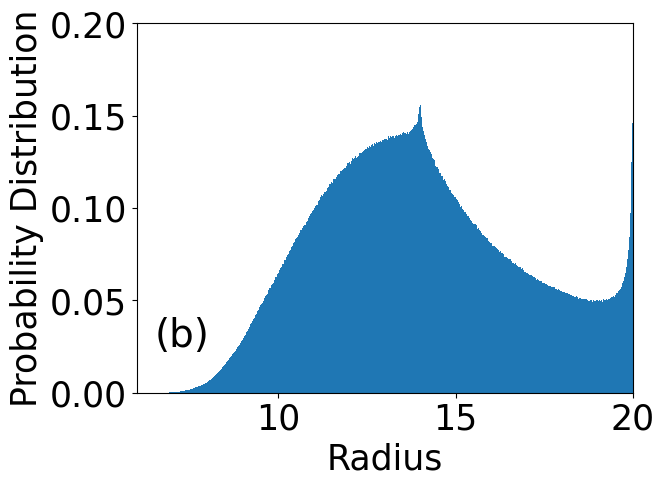}
    \includegraphics[width=0.32\linewidth]{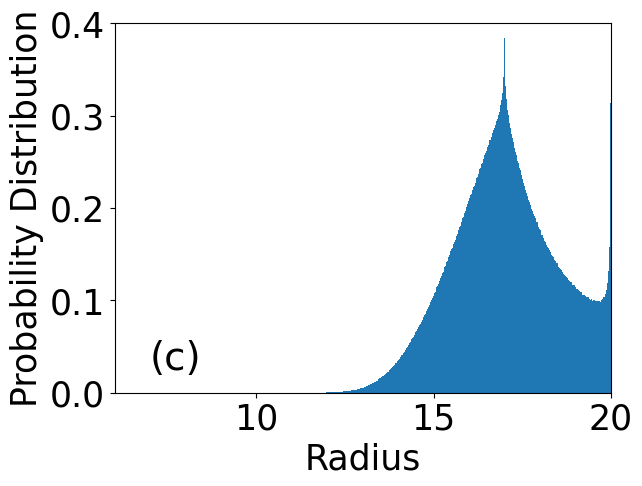}
    \includegraphics[width=0.32\linewidth]{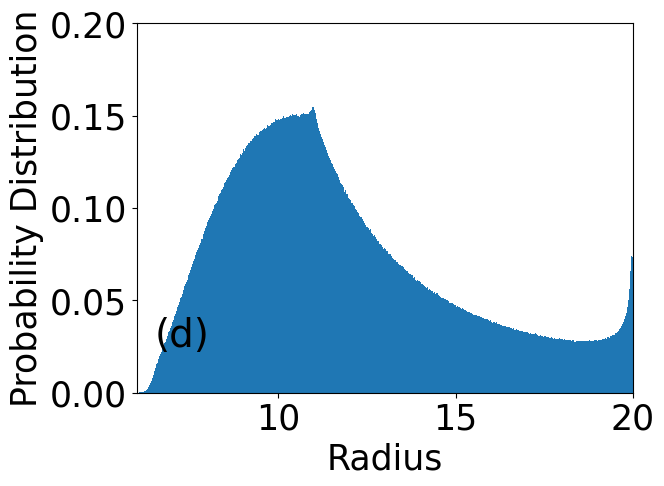}
    \includegraphics[width=0.32\linewidth]{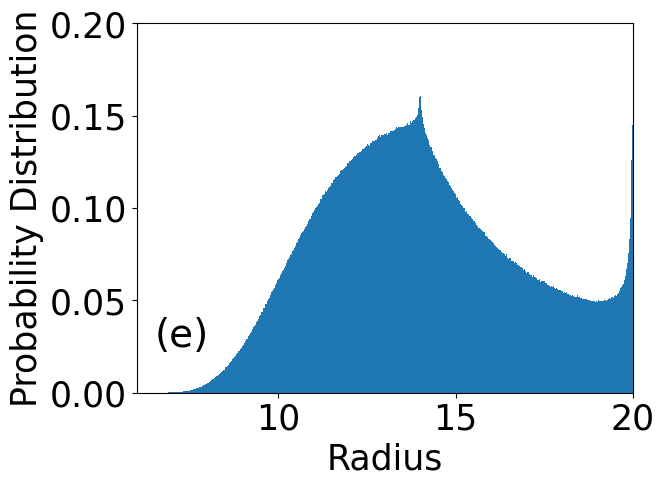}
    \includegraphics[width=0.32\linewidth]{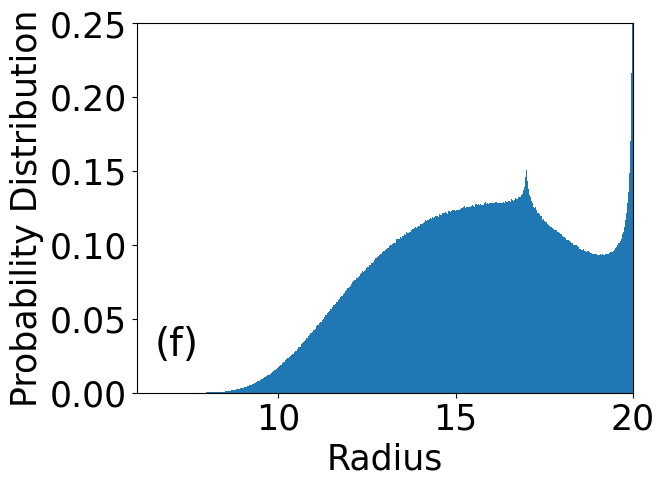}
    \caption{Radial probability distribution of reconstructed DNA. Averaged over 100,000 samples of the reconstructed curve of length 2000 $nm$. The parameters are $C_1 = 5$ and $C_3=0.3$. (a)(b)(c): Helical angle $\psi_0$ is obtained from the case $\alpha=0.001$, $\omega^2=0$. (d)(e)(f): Helical angle $\psi_0$ is obtained from the case $\alpha=0.1$, $\omega^2=0.3$, $\mu=0$. Initial position: (a)(d) $\vec{r}_0=(11,0,0)$; (b)(e) $\vec{r}_0=(14,0,0)$; (c)(f) $\vec{r}_0=(17,0,0)$.}
    \label{density_distribution}
\end{figure}
Figure \ref{density_distribution} shows the radial probability distribution of the reconstructed DNA, corresponding to initial condition $r_0 = 11, 14$ and $17$. In all cases, the distributions exhibit a peak at $r = 20$, corresponding to the capsid wall, as well as a secondary  peak near the initial position. These distributions   are consistent with the layered structure of DNA observed experimentally  in bacteriophages (e.g. \cite{Comolli2008,Lander2006}).

\subsection{Knotted DNA Configuration}
There are multiple layers of DNA filament in the capsid, for each layer, we reconstruct the DNA as described in Section \ref{reconstruction_one_layer}, with initial condition, $(r,\theta,z) = (r_i,0,0)$  for Eq. \eqref{Trajectory_Equations},
and $\beta = 0$, $\phi = \psi_0(r_i)$ for Eq. \eqref{LangevinNoise}. Here $r_i$ is the radius determined from the experimental data, $\psi_0$ is the solution obtained in Section \ref{numerical_solution_pde}.




To ensure consistency in  the DNA knotting probability calculations, we fix the total length of the DNA by truncating the height $H$ for all  layers.  The height of the i-th layer,  $ H_i=\max_j z_i^j - \min_j z_i^j$, measures the difference between the maximal and minimal $z$-coordinates of all points on that  layer.

The generation of a knotted molecule requires a single DNA strand that expands across multiple layers. However,  the proposed algorithm initially constructs separate layers. To address this, we introduce an algorithm to merge the individual layers into a single continuous trajectory.  Let $\vec{n}_{i}$ and $\vec{n}_{i+1}$ denote two consecutive layers.  We rotate the interior layer $\vec{n}_{i+1}$ until the angle $\theta$ of its first point aligns with  the angle of  the last point of the exterior layer. That is, for $i=1,\dots,N-1,$ we rotate $\vec{n}_{i+1}$ so that $\theta(\vec{n}_{i+1}^1)=\theta(\vec{n}_i^{M_i})$. Once the angles have been matched,  we connect the first strand of $\vec{n}_{i+1}$ to the last strand of $\vec{n}_i$ using linear interpolation. The full  helical trajectory is closed by projecting the open ends of the first and last layers  onto the top and bottom surfaces of a larger enclosing cylinder, and then connecting these projection points along a path on the cylinder's surface. This procedure prevents the introduction of  additional crossings when linking the two ends of the single DNA molecule. 

Once the trajectory is closed (circularized), we computed its HOMFLY-PT polynomial \cite{rawdon2023using} to identify the knot type of each trajectory. The HOMFLY-PT polynomial is a  two-variable Laurent polynomial that generalizes both the Alexander \cite{alexander1928topological} and Jones \cite{jones1997polynomial} polynomials, and provides greater accuracy  for distinguishing knot types with up to 16 crossings. Knot identification was carried out using the software Knotplot \cite{Scharein2024}.   

Figure  \ref{fig:3layers: randomized_dna} (a) shows the minimizer structure obtained by solving the corresponding Euler-Lagrange equations for three layers. Each layer is indicated with a different color with the golden layer being the closest to the capsid and the red layer being the closest to the disordered region. 
Figure \ref{fig:3layers: randomized_dna} (b)-(h) show some examples of the randomized DNA trajectories. The knot types are indicated below each figure.  The figures illustrate how the values of the parameters $C_1$ and $C_3$ of the Langevin simulation allow for the crossing of strands and therefore the generation of knots, without completely disrupting the structure of the minimizer.

\begin{figure}[!htbp]
    \centering
    \begin{subfigure}[b]{0.2\textwidth}
        \centering
        \includegraphics[width=\textwidth]{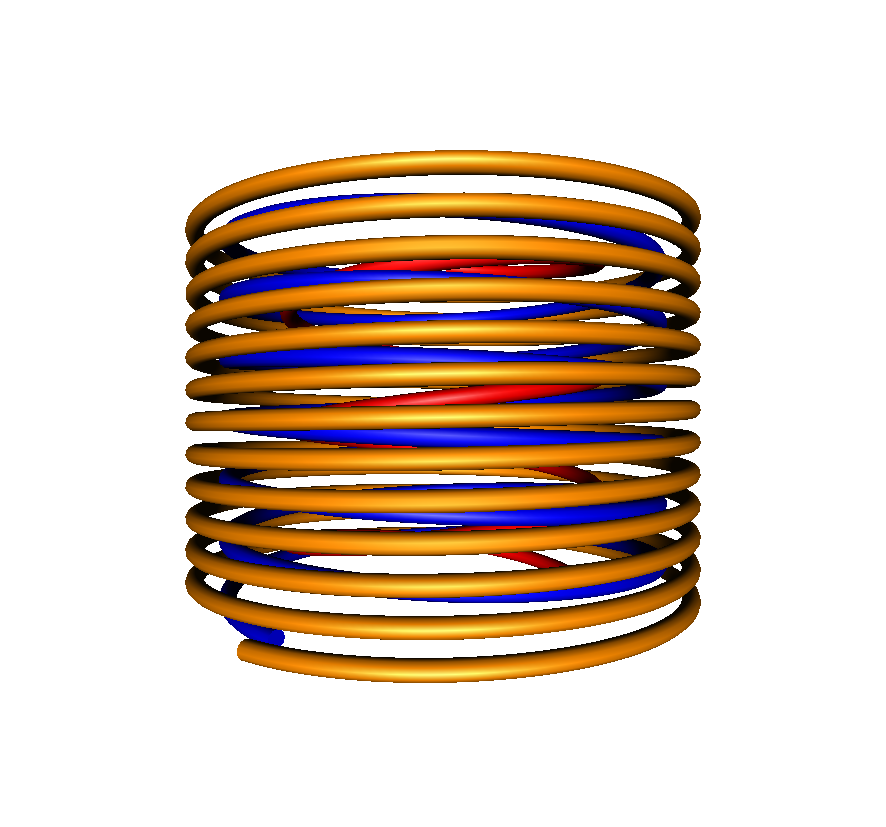}
        \caption{Minimizer}
    \end{subfigure}
    \hfill
    \begin{subfigure}[b]{0.2\textwidth}
        \centering
        \includegraphics[width=\textwidth]{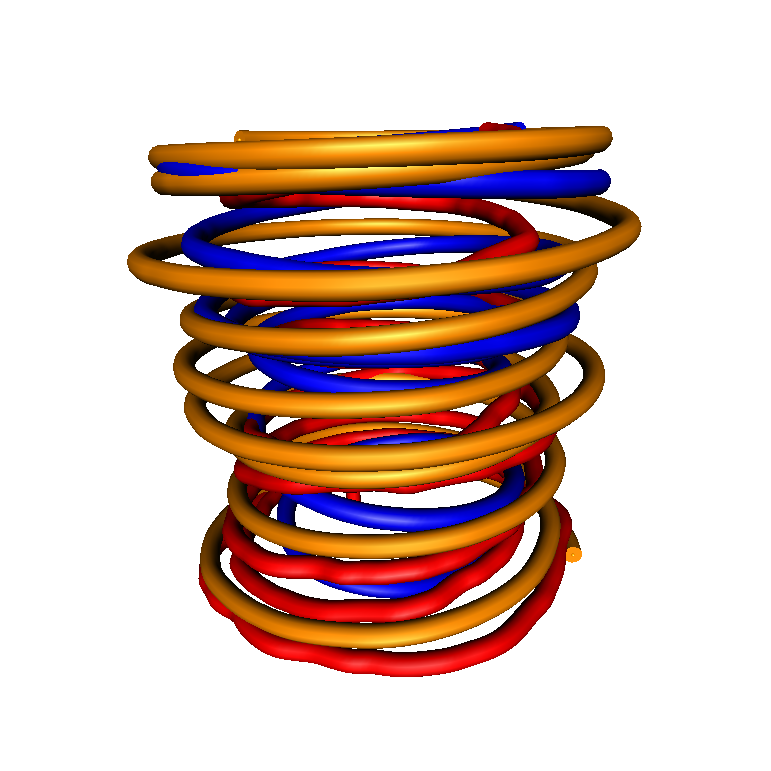}
        \caption{Unknot}
    \end{subfigure}
    \hfill
    \begin{subfigure}[b]{0.2\textwidth}
        \centering
        \includegraphics[width=\textwidth]{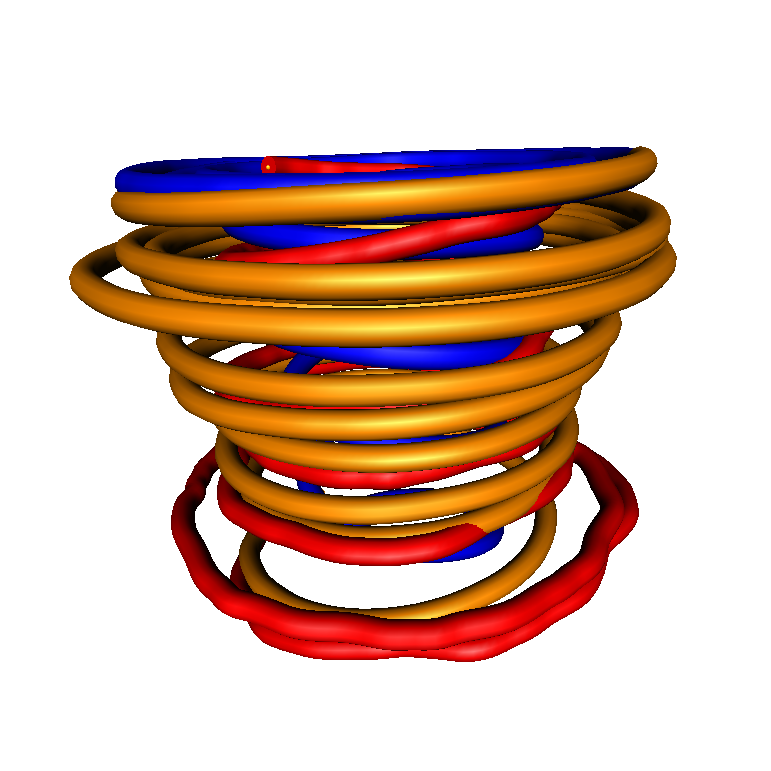}
        \caption{Trefoil}
    \end{subfigure}
    \hfill
    \begin{subfigure}[b]{0.2\textwidth}
        \centering
        \includegraphics[width=\textwidth]{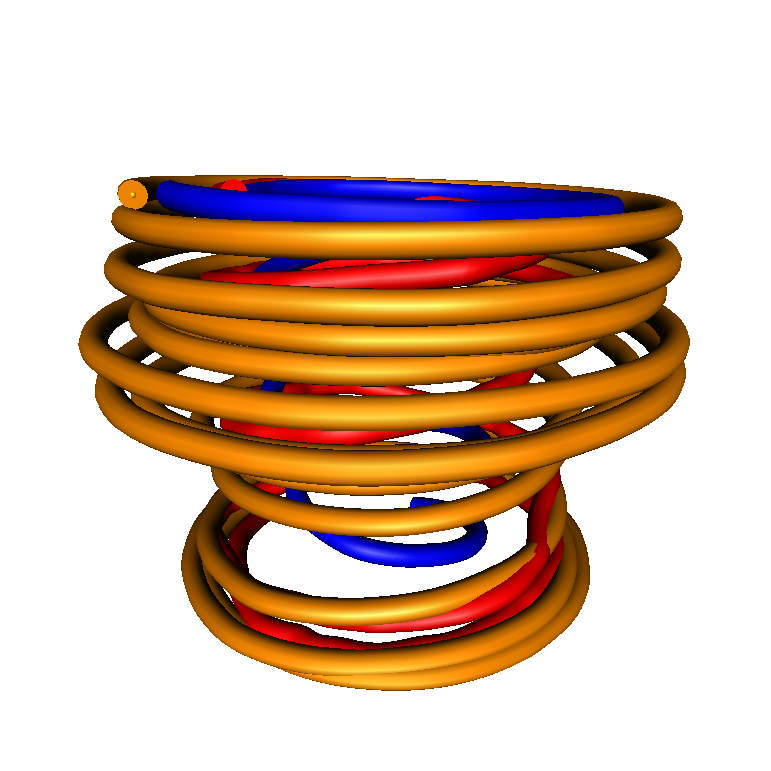}
        \caption{$4_1$}
    \end{subfigure}
    \begin{subfigure}[b]{0.2\textwidth}
        \centering
        \includegraphics[width=\textwidth]{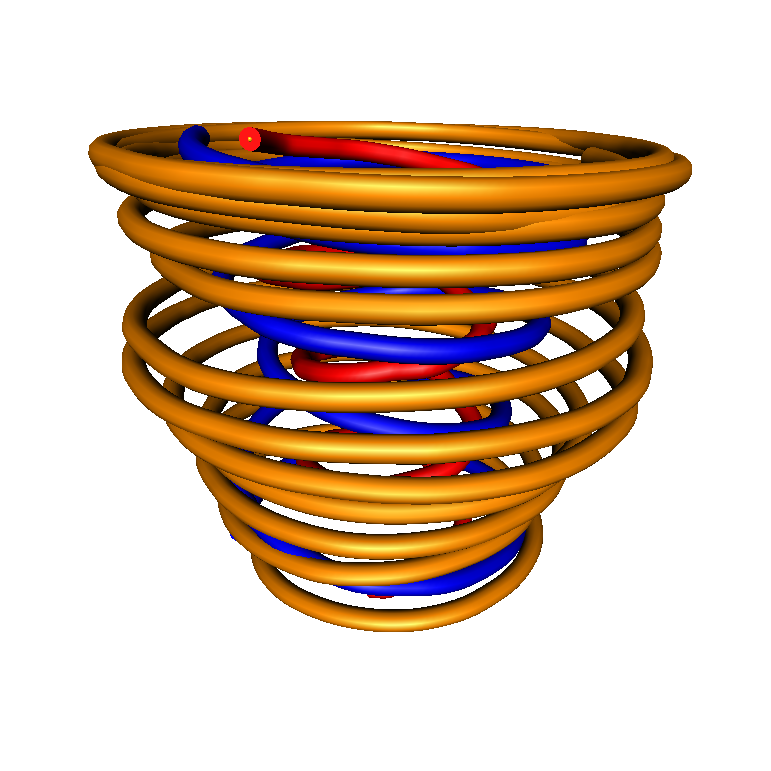}
        \caption{$5_1$}
    \end{subfigure}
    \hfill
    \begin{subfigure}[b]{0.2\textwidth}
        \centering
        \includegraphics[width=\textwidth]{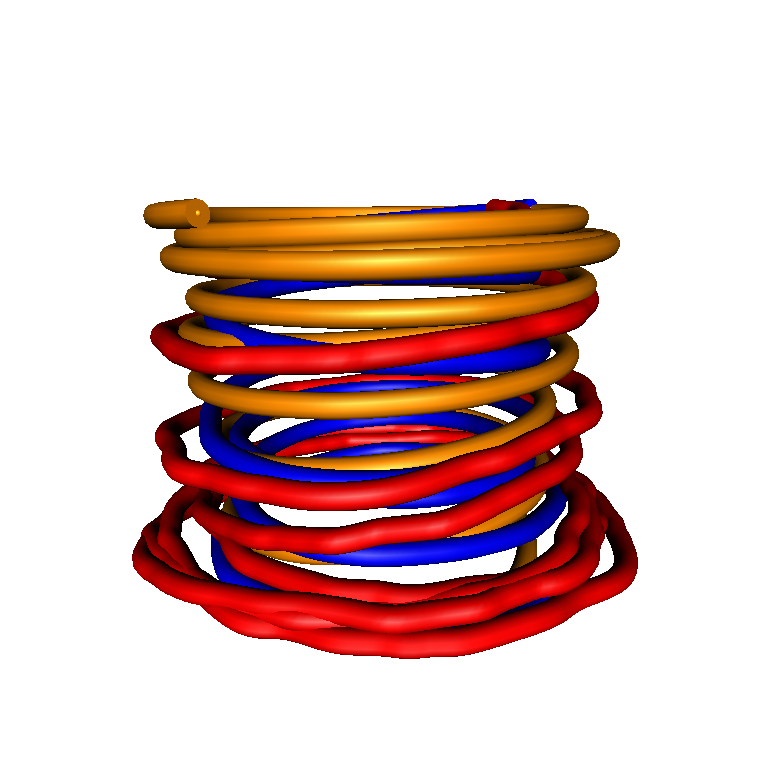}
        \caption{$5_2$}
    \end{subfigure}
    \hfill
    \begin{subfigure}[b]{0.2\textwidth}
        \centering
        \includegraphics[width=\textwidth]{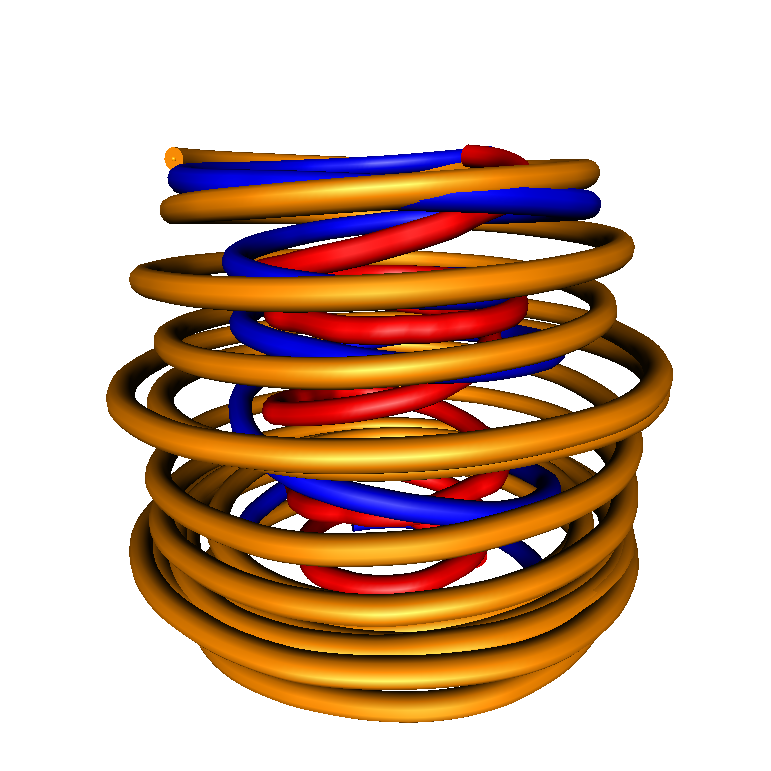}
        \caption{$7_1$}
    \end{subfigure}
    \hfill
    \begin{subfigure}[b]{0.2\textwidth}
        \centering
        \includegraphics[width=\textwidth]{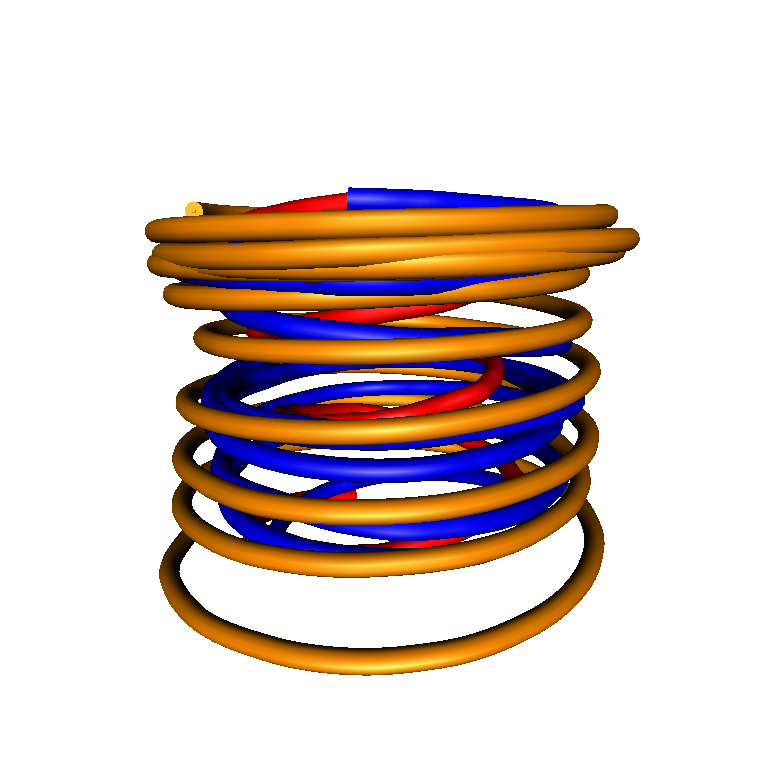}
        \caption{$3_1\# 3_1$}
    \end{subfigure}
    \caption{Examples of Reconstructed three-layer DNA trajectories. Figures generated using KnotPlot \cite{Scharein2024}.}
    \label{fig:3layers: randomized_dna}
\end{figure}

\begin{figure}[!htbp]
    \centering
    \includegraphics[width=0.7\linewidth]{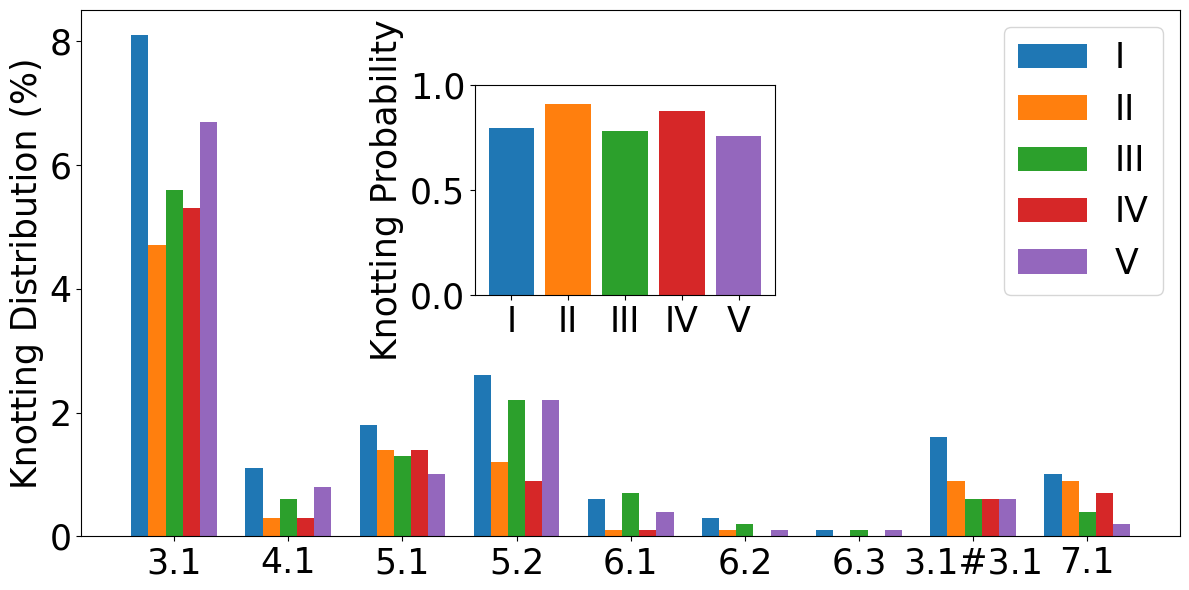}
    \caption{Knotting probability and distribution from 5000 samples of reconstructed DNA trajectories.
}
    \label{fig:placeholder}
\end{figure}


Next, we compute the knot probabilities and distributions. Results are shown in Figure \ref{fig:placeholder}. All simulations used the values estimated for P4 bacteriophage: $R_1=10\ nm$, $R_2=20\ nm$, with the first layer located at $r=17\ nm$ and total DNA length of $2200\ nm$. We consider the following 5 sets of model parameters from Figure \ref{solution_PDE}: 
\begin{align*}
    \text{Model I: } & \alpha = 0.001,\; \omega^2 = \mu = 0, 
    & \quad \text{Model II: } & \alpha = 0.1,\; \omega^2 = 0.3,\; \mu = 0, \\
    \text{Model III: } & \alpha = 0.1,\; \omega^2 = 0.3,\; \mu = -2, 
    & \quad \text{Model IV: } & \alpha = 2,\; \omega^2 = 7,\; \mu = 0, \\
    \text{Model V: } & \alpha = 2,\; \omega^2 = 6,\; \mu = -2
\end{align*}
The randomization parameters are $C_1=7$, $C_2=0.7$, $C_3=0.3$. All parameters gave knotting probabilities above $0.7$ consistent with experimental data \cite{arsuaga2002investigation}, with Models II and IV being the ones that best approximate the observed knotting probability in tailless mutants \cite{Arsuaga2005}.  The number of layers and the knot distributions also allow us to differentiate between the five models. Models III and V, both with $\mu=-2$ produced 6 layers, a value that is not consistent with other computational predictions \cite{Petrov2011}.  Models II and IV produced three layers and Model I produced four layers.   
The knot distributions show that the embeddings produced by any of the models are not random since random embeddings show a monotonic decrease of the knot frequency with an increasing knot crossing number. In particular, random embeddings show that the four crossing knot population is always larger than the populations of the five crossing knots \cite{Arsuaga2005,deguchi1997universality}. The proportion between the toroidal and twist five crossing knot is also very informative and helps us differentiate between the models. Experimental data show that the toroidal knot $5_1$ is more frequent than the twist knot $5_2$. In this case, models I, III and IV are not consistent with those observations but models II and IV are. The distribution of populations of knots with larger crossing number larger than five has not been fully characterized experimentally but experimental data clearly shows two populations for six and seven crossings knots. In both cases, one population is larger than the other. In our simulations, we observed that the sum of trefoils is the most prevalent, followed by $6_1$, with small amounts of $6_2$ and $6_3$.  From this analysis we conclude that the different embeddings generated by the parameters $\mu$, $\omega^2$ and $\alpha$ produce very distinct knotting distributions that can be used to further refine the analysis of experimental data. Performing a full scale analysis that relates the parameters of the model with the knotting distributions, however, is beyond the scope of this paper  and will be reported elsewhere.

\section{Discussion}
The properties of the DNA molecule in confined geometries  remain to be fully understood. The bacteriophage is an excellent model to study these properties because it packs its micron-long genome in a capsid that is a few tenths of nanometers long. These packing conditions have been shown to introduce topological changes in the DNA molecule in the form of knots \cite{arsuaga2002knotting,liu1981knotted,Liu1981,Wolfson1985} and to induce DNA chromonic liquid crystal phases. While random knotting models such as \cite{Arsuaga2005,arsuaga2002knotting, arsuaga2008,diao2014knot} have been used to study topological changes of DNA, and continuum models of liquid crystals using the Oseen-Frank theory, such as \cite{ortiz2003,hiltner2021chromonic,walker2020fine,liu2022helical} have been proposed to describe the free energy of DNA inside viral capsids, to our knowledge no mathematical model has incorporated the two.  In this work, we introduce a free energy formulation for DNA inside bacteriophages that takes into consideration the twist, bending and saddle-splay terms from the Oseen-Frank energy for liquid crystals modulated by a DNA density factor. Our proposed free energy also considers the role of entropy and the chemical potential of the DNA molecule. Assuming that the DNA molecule follows a helical trajectory away from north and south caps of the capsid (as evidenced by cryoEM observations \cite{Comolli2008,Lander2006}) we show that the free energy has a unique minimizer. We also show how the ratio between the Frank constants ($\alpha=\frac{k_3}{k_2}$, and the chemical potential ($\mu$) determine the packing organization of DNA. Since the minimizer obtained using the Oseen-Frank theory gives an unknotted configuration, we have implemented a Langevin simulation that generates knotted conformations that preserve the overall structure of the minimizer and that give knotting probability values similar to those observed experimentally. Interestingly, we observed that the different embeddings generated by the parameters of the model $\mu,\omega^2,\alpha$ generate different number of layers and knot distributions and therefore they can be used to analyze experimental data \cite{Arsuaga2005}. 
The proposed energy in Eq. \ref{S1} can be further explored. The role of the chemical potential $\mu$ is particularly interesting since the packed DNA length changes as as a function of $\mu$ without imposing any further conditions on the DNA density. Questions previously addressed in other works \cite{hiltner2021chromonic,walker2020fine}, such as the size ratio between ordered and disordered regions, can also be addressed by implementing the natural boundary conditions. In addition, the energy term associated with the Frank constant, $K_2 + K_4$, acts as a penalty that governs the energetic cost of the transition between these regions.


\bibliographystyle{plain}
\bibliography{gel, pei, chromonic}
\end{document}